\documentclass[11pt]{article}
\usepackage{geometry}
\usepackage{bm}
\usepackage{float}
\usepackage{color}
\usepackage{graphicx}
\usepackage{amssymb, amsmath, amsfonts}
\usepackage{bbm}
\usepackage{amsthm}

\usepackage{comment}

\usepackage{enumitem}
\usepackage{hyperref}

\renewenvironment{proof}[1][]
    {\noindent
       \ifx&#1&{\it Proof.}
       \else{\it Proof ({#1}).}
       \fi}{\hfill $\blacksquare$}

\newtheorem{theorem}{Theorem}[section]
\newtheorem{claim}[theorem]{Claim}
\newtheorem{remark}[theorem]{Remark}

\newtheorem{proposition}[theorem]{Proposition}
\newtheorem{definition}[theorem]{Definition}
\newtheorem{corollary}[theorem]{Corollary}
\newtheorem{lemma}[theorem]{Lemma}
\newtheorem{observation}[theorem]{Observation}

\DeclareMathOperator*{\E}{\mathbb{E}}

\DeclareMathOperator{\HH}{H}
\DeclareMathOperator{\I}{I}

\DeclareMathOperator{\KL}{D_{\text{KL}}}
\DeclareMathOperator{\Supp}{\text{Supp}}
\DeclareMathOperator{\Ent}{\text{Ent}}

\newcommand{\SECS}{{\mathcal{D}}}
\newcommand{\secv}{{d}}
\newcommand{\SECV}{{D}}

\newcommand{\NN}{\mathbb{N}}
\newcommand{\RR}{\mathbb{R}}
\newcommand{\Qcal}{\mathcal{U}}
\newcommand{\calO}{\mathcal{O}}

\newcommand{\SUC}{\boldsymbol{suc}}
\newcommand{\TR}{\boldsymbol{tr}}
\newcommand{\POST}{\boldsymbol{post}}
\newcommand{\PC}{\boldsymbol{PC}}
\newcommand{\MID}{\boldsymbol{MID}}
\newcommand{\INP}{I}
\newcommand{\OUTP}{O}
\newcommand{\MAXS}{\boldsymbol{MaxS}}

\newcommand{\ind}[1]{\binom{[#1]}{1}}

\title{Non-Adaptive Cryptanalytic Time-Space Lower Bounds via a Shearer-like Inequality for Permutations}

\author{Itai Dinur\thanks{Ben-Gurion University and Georgetown University. \texttt{dinuri@bgu.ac.il}. Supported by a gift to Georgetown University.}
\and Nathan Keller\thanks{Bar-Ilan University. \texttt{nathan.keller27@gmail.com}. Partially supported by the Israel Science Foundation (grants no.\ 2669/21 and 2456/25) and by the US-Israel Binational Science Foundation (grant no.\ 2024120).}
\and Avichai Marmor\thanks{Bar-Ilan University. \texttt{avichai@elmar.co.il}. Partially supported by the Israel Science Foundation (grants no.\ 2669/21 and 2456/25) and by the Bar Ilan University's Presidential Scholarship.}}
\date{}

\begin{document}

\maketitle

\begin{abstract}
    The power of adaptivity in algorithms has been intensively studied in diverse areas of theoretical computer science. In this paper, we obtain a number of sharp lower bound results which show that adaptivity provides a significant extra power in cryptanalytic time-space tradeoffs with (possibly unlimited) preprocessing time.  

    Most notably, we consider the discrete logarithm (DLOG) problem in a generic group of $N$ elements. The classical `baby-step giant-step' algorithm for the problem has time complexity $T=O(\sqrt{N})$, uses $O(\sqrt{N})$ bits of space (up to logarithmic factors in $N$) and achieves constant success probability.

    We examine a generalized setting where an algorithm obtains an advice string of  $S$ bits and is allowed to make $T$ arbitrary non-adaptive queries that depend on the advice string (but not on the challenge group element for which the DLOG needs to be computed). 
    
    We show that in this setting, the $T=O(\sqrt{N})$ online time complexity of the baby-step giant-step algorithm cannot be improved, unless the advice string is more than $\Omega(\sqrt{N})$ bits long. This lies in stark contrast with the classical adaptive Pollard's rho algorithm for DLOG, which can exploit preprocessing to obtain the tradeoff curve $ST^2=O(N)$. We obtain similar sharp lower bounds for the problem of breaking the Even-Mansour cryptosystem in symmetric-key cryptography and for several other problems. 
    
    To obtain our results, we present a new model that allows analyzing non-adaptive preprocessing algorithms for a wide array of search and decision problems in a unified way. 
    
    Since previous proof techniques inherently cannot distinguish between adaptive and non-adaptive algorithms for the problems in our model, they cannot be used to obtain our results. Consequently, we rely on information-theoretic tools for handling distributions and functions over the space $S_N$ of permutations of $N$ elements. Specifically, we use a variant of Shearer's lemma for this setting, due to Barthe, Cordero-Erausquin, Ledoux, and Maurey (2011), and a variant of the concentration inequality of Gavinsky, Lovett, Saks and Srinivasan (2015) for read-$k$ families of functions, that we derive from it. This seems to be the first time a variant of Shearer's lemma for permutations is used in an algorithmic context, and it is expected to be useful in other lower bound arguments.
\end{abstract}

\thispagestyle{empty}

\clearpage
\setcounter{page}{1}

\section{Introduction}

\subsection{Background}

\paragraph{Cryptanalytic time-space tradeoffs.}
We consider cryptanalytic time-space tradeoffs in the \emph{preprocessing} (i.e.,~non-uniform) model.
An algorithm in this model is divided into two phases. In the first phase, the preprocessing algorithm obtains access to an input defined by the setting of the problem (e.g., $f: [N] \to [N]$ in the function inversion problem), and produces an `advice string' of $S$ bits. In the second phase, the online algorithm receives a specific challenge it has to solve with respect to the given input (e.g., an output of the function to invert).
The online algorithm is given the advice string, along with oracle access to the input. The complexity of the algorithm is evaluated by the bit-length $S$ of the advice and the time complexity $T$ of the online phase, typically measured in terms of the number of oracle queries. 
Such an algorithm is called an $(S,T)$-algorithm.

The preprocessing model is very natural from a practical point of view, as in many scenarios, the adversary is willing to solve many specific instances of the problem and the cost of a one-time preprocessing is amortized over the multiple solved instances. 

The first and most famous cryptanalytic algorithm in the preprocessing model is Hellman's algorithm~\cite{Hellman80} for inverting a random function.
For constant success probability, Hellman's algorithm obtains a time-space tradeoff of about $T S^2 = \tilde{O}(N^2)$ (where $\tilde{O}$ hides logarithmic factors in $N$) on a domain of size $N$. The algorithm was extended for inverting any function by Fiat and Naor~\cite{FiatN99}.

In recent years, there has been a significant effort in the cryptographic community to understand the power of non-uniform algorithms by devising new algorithms and proving time-space lower bounds against them. The emphasis has been on proving lower bounds for various problems (see, for example,~\cite{DodisGK17,CorettiDG18,CorettiDGS18,Corrigan-GibbsK18,Corrigan-GibbsK19,BartusekMZ19,GolovnevGHPV20,ChawinHM20,ChungGLQ20,AkshimaCDW20,GuoLLZ21,FreitagGK22,GolovnevGPS23,AkshimaBGXY24,AkshimaGL24,GajulapalliGK24,AkshimaBCGY25}).

\paragraph{Non-adaptive cryptanalytic time-space tradeoffs.}
An algorithm with oracle access to an input is called \emph{adaptive} if its oracle queries may depend on the outcome of its previous queries. Otherwise it is called non-adaptive.
In this paper, we will be interested in non-adaptive cryptanalytic algorithms receiving advice. 

Almost all best-known cryptanalytic time-space tradeoff algorithms are adaptive. This includes (for example) Hellman's algorithm~\cite{Hellman80} and the Fiat and Naor~\cite{FiatN99} algorithm for function inversion, and the algorithms by 
Mihalcik~\cite{Mih10}, Bernstein and Lange~\cite{BernsteinL13} and Corrigan-Gibbs and Kogan~\cite{Corrigan-GibbsK18} for discrete log (DLOG).

However, in most cases there is no proof that non-adaptive algorithms cannot perform as well as the best known adaptive ones. For example, for the function inversion problem, the best-known lower bound of $ST \geq \tilde{\Omega}(N)$ is far from the best-known tradeoff of $T S^2 = \tilde{O}(N^2)$. Yet, for non-adaptive algorithms, the gap is even larger, as the lower bound remains roughly $ST \geq \tilde{\Omega}(N)$ (see~\cite{GajulapalliGK24}), but the best-known algorithm 
cannot do better than $\max(T,S) = \tilde{O}(N)$.

In 2019, Corrigan-Gibbs and Kogan~\cite{Corrigan-GibbsK19} formally raised the question of the power of adaptivity in the context of the function inversion problem. 
The authors of~\cite{Corrigan-GibbsK19} argued that the lack of progress in establishing such lower bounds can be explained by a classical barrier in complexity theory. Indeed, they showed that a significantly improved lower bound on strongly non-adaptive algorithms for function inversion would imply a circuit lower bound against Boolean circuits of linear size and logarithmic depth, thus resolving a 50-year old well-known open question of Valiant~\cite{Valiant77}. Later on, it was shown that such a lower bound would also imply significant data structure lower bounds~\cite{DvorakKKS21,GolovnevGHPV20}. Following the work of Corrigan-Gibbs and Kogan~\cite{Corrigan-GibbsK19}, the power of adaptivity in the context of the function inversion problem was studied in a number of papers~\cite{ChawinHM20,DvorakKKS21, GajulapalliGK24,GolovnevGPS23}. 

\paragraph{Our work.} In this paper, we investigate the question of~\cite{Corrigan-GibbsK19} in a more general context:
\begin{center}
    What is the power of adaptivity in non-uniform cryptanalytic algorithms?
\end{center}

We present sharp lower bound proofs which demonstrate the significant power of adaptivity in a variety of cryptographic problems, in the preprocessing model.
More specifically, in the context of public-key cryptography, we present sharp non-adaptive time-space lower bounds for the generic DLOG problem and the square-DDH (decisional Diffie-Hellman) problem. We also prove a  
lower bound for DDH that is sharp in a part of the range. In the context of symmetric-key cryptography, we present sharp time-space lower bounds for key-recovery attacks against the classical Even-Mansour cipher. 

From a technical viewpoint, we show in Section~\ref{sec:proof-overview} and Appendix~\ref{app:limitations} that previous lower bound proof techniques for the problems we consider are inherently limited and cannot distinguish between adaptive and non-adaptive algorithms.  

\paragraph{Definition of adaptivity in our model.}
In the problems we analyze, each query of the online adversary is first specified in a \emph{pre-translated} form and then processed by a translation function before it is passed to the oracle.
Our notion of non-adaptivity allows the pre-translated queries of the online algorithm to depend on the advice string, but not on the challenge itself.

Importantly, after translation, the resulting oracle queries \emph{may depend on the challenge}. For example, in the DLOG setting (see details below), the challenge is an element $x=g^d$ for an unknown exponent $d$, and a pre-translated query is specified by a pair $(a,b)$ chosen by the adversary. This pair is then converted into the group element $x^a g^b = g^{ad+b}$, so although $(a,b)$ itself is chosen independently of the challenge, the actual oracle query still depends on the hidden value $d$.

The implicit dependence on the challenge allows our notion of non-adaptivity to capture classical algorithms such as baby-step giant-step for DLOG (see below). This is in contrast to the function inversion problem, where requiring the queries to be independent of the challenge trivializes the problem. Our lower bounds show that this class of non-adaptive algorithms cannot be improved even if we allow an arbitrary advice string (of length $S$ bits), arbitrary non-adaptive queries (that may depend on the advice), and unlimited preprocessing time.

As we argue towards the end of the Introduction, extending our proof to handle challenge-dependent pre-translated queries seems to require proving new types of concentration inequalities. 

Before presenting our results, we elaborate on the rich history of the problems we analyze.

\paragraph{The generic discrete logarithm problem.} Given a group $G$ of $N$ elements with a fixed generator $g$, the DLOG problem in $G$ asks for solving the equation $g^d=x$, for a given $x \in G$. 
It is widely believed that the problem is hard if the group is chosen carefully, and the security of various cryptosystems relies on the hardness of this problem and its variants.
In 1997, Shoup~\cite{Shoup97} presented the \emph{generic group model} (GGM), intended for proving lower bounds against a specific class of algorithms that are applicable to any group of a particular order.
A major motivation for studying this model is that in some cryptographically relevant groups (such as some elliptic curve groups),
the best-known algorithms for DLOG are 
generic. 
%algorithms are the best-known algorithms for DLOG.

The model uses an auxiliary set $\mathcal{W}$ of bit strings with $|\mathcal{W}|\geq N$. In the model, the elements of $G$ are `hidden' by mapping $g^j$ to $\sigma(j)$, where $\sigma:\mathbb{Z}_N \to \mathcal{W}$ is a randomly chosen bijective function. (We show in the sequel that there is no loss of generality in stating our results with $\mathcal{W}=[N] = \{1,2,\ldots,N\}$, and so we state them this way.) An algorithm for solving the problem receives as inputs the encodings $\sigma(1),\sigma(d)$ of $g,x$, and its goal is to find $d$. The oracle queries allowed for the algorithm are of the form $\sigma(a \cdot d+b \bmod N)$.\footnote{The original model of Shoup is slightly different but it is easy to show that this does not affect the analysis.} 
The complexity of the algorithm is the total number of queries it makes, and its success probability is the expectation of the probability that it outputs $d$, over all uniform choices of $\sigma,d$, and the internal randomness of the algorithm. In this model, Shoup proved that if $N$ is prime, then the success probability of any algorithm for DLOG that makes $T$ queries is at most $T^2/N$.

The $\Omega(\sqrt{N})$ lower bound asserted by Shoup on the complexity of any algorithm for DLOG in the GGM is matched by the classical \emph{baby-step giant-step} algorithm proposed by Shanks~\cite{Shanks71} in 1971. The algorithm is based on writing the equation $g^d=x$ in the form $g^{im+j}=x$, where $m=\lceil \sqrt{N} \rceil$ and $0 \leq i,j < m$, which in turn can be written as $g^j=x(g^{-m})^i$. The algorithm computes $g^j$ for $1 \leq j \leq m$ and stores the values in a table. Then, it computes $y=g^{-m}$ and tries to find matches between $x, xy, xy^2, \ldots$ and values in the table. Once a match of the form $g^j=x(g^{-m})^i$ is found, we know that $d=im+j$ is the solution of the equation $g^d=x$, that is, the discrete logarithm of $x$. The algorithm can be executed in the generic group model, by replacing each computation of $g^j$ by the query $\sigma(j)$ and each computation of $xy^i$ by the query $\sigma(im+d \bmod N)$. In particular, this shows that baby-step giant-step is non-adaptive in our sense: the online algorithm fixes in advance, independently of the challenge, the coefficient pairs of the form $(0,j)$ and $(1,-im)$, that yield the required oracle queries. The algorithm uses memory $O(\sqrt{N})$ words of $O(\log N)$ bits, and its query complexity is $O(\sqrt{N})$. 

In 1978, Pollard~\cite{Pollard78} introduced the \emph{Pollard Rho} algorithm for DLOG, which obtains the same query complexity of $O(\sqrt{N})$ while using only a very small amount $\tilde{O}(1)$ of memory. The algorithm is based on finding collisions of the form $x^a g^b = x^{a'} g^{b'}$ and using them to find a solution of the equation $g^d=x$ via the extended Euclid's algorithm. The collisions are found using Floyd's cycle finding algorithm~\cite{Knuth69}, which finds a collision in a random function $f:[N] \to [N]$ in time $O(\sqrt{N})$, using only a few memory cells. This algorithm can also be executed in the GGM, as computations of $x^a g^b$ correspond to queries of the form $\sigma(a \cdot d + b \bmod N)$, and the rest of the algorithm is based on comparisons.

The baby-step giant-step algorithm is clearly non-adaptive. In contrast, Pollard’s Rho algorithm is highly dependent on adaptivity using Floyd's algorithm, where each query depends on the output of the previous one.

\paragraph{Time-space tradeoffs for the generic DLOG problem with preprocessing.} The generic group model is naturally extended to algorithms with preprocessing. In this case, an algorithm $A$ consists of a `preprocessing' algorithm $A_0$ and an `online' algorithm $A_1$. The algorithm $A_0$ has direct access to the encoding function $\sigma$ (but not to the input $\sigma(d)$) and produces an advice string $z$. 
The algorithm $A_1$ is defined as in the standard generic group model described above, but also receives the advice string $z$ as input.

The baby-step giant-step algorithm can be naturally viewed as an algorithm with preprocessing, as its first stage of computing $\{g^j:0 \leq j < m\}$ can be done before the value of $x$ is known. However, it does not seem to fully exploit the power of preprocessing, since the time complexity of the preprocessing phase %(which is allowed to be unbounded) 
is equal to its space complexity.

Pollard's Rho algorithm does not use preprocessing. Yet, as noted above, 
Mihalcik~\cite{Mih10}, Bernstein and Lange~\cite{BernsteinL13} and Corrigan-Gibbs and Kogan~\cite{Corrigan-GibbsK18} presented generic algorithms that solve DLOG with preprocessing with high probability whose time and space complexity satisfy $ST^2 \leq \tilde{O}(N)$. In particular, if space of $\tilde{O}(N^{1/3})$ is allowed, then the online time complexity can be reduced to $O(N^{1/3})$. These algorithms are known to be best possible~\cite{Corrigan-GibbsK18}. 

\paragraph{The Decisional Diffie Hellman (DDH) and the square-DDH (sqDDH) problems.} The DDH problem in the GGM asks to distinguish between the triples $(\sigma(x),\sigma(y),\sigma(xy))$ and $(\sigma(x),\sigma(y),\sigma(z))$, where $x,y,z \in G$ are random elements. This corresponds to the situation in the Diffie-Hellman key exchange scheme (and many other cryptosystems) where the adversary sees the elements $g^x,g^y$ sent between the parties and wants to decide whether a candidate element is the shared secret key $g^{xy}$ or a random element $g^z$. When introducing the GGM, Shoup~\cite{Shoup97} showed that for a prime $N$, any algorithm for the DDH problem in the GGM with $|G|=N$ that makes $T$ queries, has a success probability of at most~$\frac{1}{2} +\frac{T^2}{N}$.

The sqDDH problem~\cite{MaurerW99} is the special case of DDH in which $x=y$. Namely, it asks to distinguish between the pairs $(\sigma(x),\sigma(x^2))$ and $(\sigma(x),\sigma(z))$, where $x,z \in G$ are uniform elements. While there are groups in which sqDDH is significantly easier than DDH (see~\cite{JouxN03}), it can be shown that in the GGM, any algorithm for sqDDH that makes $T$ queries, has a success probability of at most $\frac{1}{2} +\frac{T^2}{N}$, just like for DDH (see~\cite{CorettiDG18}).

As in the case of the DLOG problem, the complexity of algorithms for DDH and sqDDH can be significantly reduced if preprocessing is allowed. Indeed, for the DDH problem with preprocessing, the generic algorithms of Mihalcik~\cite{Mih10}, Bernstein and Lange~\cite{BernsteinL13}, and Corrigan-Gibbs and Kogan~\cite{Corrigan-GibbsK18} achieve success probability of $\frac{1}{2}+\tilde{\Omega}(ST^2/N)$. For the sqDDH problem, Corrigan-Gibbs and Kogan~\cite{Corrigan-GibbsK18} presented a generic algorithm with preprocessing that achieves an even higher success probability of $\frac{1}{2}+\tilde{\Omega}(\sqrt{ST^2/N})$. In the other direction, Corrigan-Gibbs and Kogan~\cite{Corrigan-GibbsK18} showed that any algorithm with preprocessing for the DDH problem or for the sqDDH problem in the GGM has success probability of $\frac{1}{2}+\tilde{O}(\sqrt{ST^2/N})$. Interestingly, the correct tradeoff formula for the DDH problem remained open for several years, until recently Akshima, Besselman, Guo, Xie and Ye~\cite{AkshimaBGXY24} proved that it is equal to $\frac{1}{2}+ \tilde{\Theta}(ST^2/N)$ by presenting a tighter proof.

\paragraph{The Even-Mansour (EM) cryptosystem.} The EM cryptosystem was proposed in 1991 by Even and Mansour~\cite{EvenM97} as the `simplest possible' construction of a block cipher from a public permutation. It is defined as $EM(m)=k_2 \oplus \sigma(k_1 \oplus m)$, where $k_1,k_2 \in \{0,1\}^n$ are $n$-bit keys and $\sigma:\{0,1\}^n \to \{0,1\}^n$ is a publicly known permutation. The EM cryptosystem has become a central element in block cipher constructions and the security level of EM and of its iterative variants was studied extensively (see, e.g.,~\cite{ChenLLSS18,CogliatiLS15,DaiSST17}). 

An adversary is allowed to make encryption/decryption queries to $EM$ and public queries to $\sigma$ or $\sigma^{-1}$.
We consider key-recovery attacks which aim at recovering $(k_1,k_2)$. Even and Mansour showed that if $\sigma$ is chosen uniformly at random, then such an attack which makes $T_2$ encryption/decryption queries to $EM$ and $T_1$ queries to $\sigma$ or $\sigma^{-1}$, has a success probability of $O(T_1 T_2 / 2^n)$. The work~\cite{DKS15} showed that the same holds for the seemingly weaker \emph{Single-key EM} cryptosystem, defined as $SEM(m)=k_1 \oplus \sigma(k_1 \oplus m)$ (i.e., EM with $k_2=k_1$).

Both adaptive and non-adaptive attacks matching these bounds are known~\cite{BW00,Daemen91,DKS15}. Like in the case of DLOG, it was shown that the online complexity can be reduced significantly in the preprocessing model. Specifically, a key recovery attack with complexity $S=T_1=T_2=\tilde{O}(2^{n/3})$ was presented in~\cite{FouqueJM14} and an upper bound of $\frac{1}{2}+ \tilde{O}(\sqrt{S(T_1+T_2)T_2/2^n}+T_1 T_2/2^n)$ on the success probability of any distinguishing attack, which matches the attack in the constant success probability setting, was shown in~\cite{CorettiDG18}. The previous works on EM are described in detail in Section~\ref{sec:EM}.  

\subsection{Our results}
\label{sec:results}

\paragraph{Main motivation.} The main motivation behind this work is seeking a better understanding of classical generic algorithms for DLOG and other problems.
Although Pollard's algorithm for DLOG is over 45 years old, we are not aware of any formal statement that justifies its need for adaptivity.
Indeed, proving such a statement seems to be related to the major open problem of proving strong time-space lower bounds against short-output oblivious branching programs,\footnote{Technically, the open problem is to prove slightly super-linear time-space tradeoff lower bounds for oblivious branching programs.
In our case, since the best algorithm runs in time $O(\sqrt{N})$, the analogous problem is to prove a time-space tradeoff of the form $TS \geq N^{1/2 + \epsilon}$ for some $\epsilon > 0$.}  for which the best-known result was obtained by Babai, Nisan and Szegedy~\cite{BabaiNS92} more than 30 years ago.
Yet, given the historical significance of Pollard's algorithm, we view the question of formally justifying its adaptivity in an alternative (but standard) computational model as fundamental. We make progress on this question by showing that adaptivity does provide a significant advantage for DLOG algorithms in the generic group model, in the \emph{preprocessing} (auxiliary input) setting. We also obtain similar results for several other major problems.

In particular, we present a new model, which we call \emph{permutation challenge} (PC), which allows proving lower bounds on non-adaptive time-space tradeoffs for a wide variety of problems, including DLOG, DDH, sqDDH, and breaking the EM cryptosystem.\footnote{It seems possible to extend the model to prove time-space lower bounds for the permutation inversion problem in case the adversary is non-adaptive and makes challenge-dependent queries of a certain type (e.g., for each challenge, the query set is selected independently uniformly at random). However, such lower bounds are unlikely to improve the state-of-the-art, hence we did not pursue this direction further.}
We describe the model as part of the proof overview in Section~\ref{sec:proof-overview}. We now describe the results we obtain for the specific problems. 

\paragraph{Sharp lower bound for non-adaptive algorithms for DLOG with preprocessing.} Our main result concerns non-adaptive generic algorithms with preprocessing for the DLOG problem. We show that unlike the adaptive setting, for non-adaptive algorithms preprocessing cannot be exploited better than in the baby-step giant-step algorithm -- indeed, its $O(\sqrt{N})$ online time complexity cannot be reduced even if a preprocessing phase of unbounded complexity is allowed, as long as the size of the advice $z$ is $\tilde{O}(\sqrt{N})$.
\begin{theorem}\label{thm:dl-nonadaptive-intro}
Let $A=(A_0,A_1)$ be a non-adaptive $(S,T)$-algorithm for the DLOG problem in the generic group model, over a group $G$ with a prime number $N$ of elements. Denote by $\MAXS_{\textup{DLOG}}(T)$ 
the optimal success probability of a non-preprocessing, non-adaptive algorithm that makes at most $T$ queries.
Then, the success probability of $A$ is at most
$$
2\cdot \MAXS_{\textup{DLOG}}(T) + \frac{4\log_e(2) ST}{N} + \frac{T^2}{N} \leq
\frac{3T^2 }{N}  + \frac{4\log_e(2) ST}{N}.
$$
\end{theorem}
The bound on $\MAXS_{\textup{DLOG}}(T) \leq \frac{T^2}{N}$ is by Shoup's theorem~\cite{Shoup97} for any non-preprocessing DLOG algorithm.

\paragraph{Lower bounds for non-adaptive algorithms for DDH and sqDDH with preprocessing.} 

Our next results analyze decisional problems in the GGM.

\begin{theorem}\label{thm:DDH-intro1}
Let $A=(A_0,A_1)$ be a non-adaptive $(S,T)$ algorithm for the DDH (resp.~sqDDH) problem in the generic group model, over a group $G$ with a prime number $N$ of elements. Denote by $\MAXS_{\textup{DDH}}(T)$ the optimal success probability of a non-preprocessing, non-adaptive algorithm that makes at most $T$ queries. Then, the success probability of $A$ is at most
$$
\MAXS_{\textup{DDH}}(T) + \sqrt{\frac{2\log_e(2) ST}{N}} + \frac{T^2}{N} \leq
\frac{1}{2}+ \frac{2T^2}{N}  + \sqrt{\frac{2\log_e(2)ST}{N}},
$$
\end{theorem}
\noindent where the inequality is by Shoup's theorem~\cite{Shoup97} which bounds $\MAXS_{\textup{DDH}}(T) \leq \frac{1}{2} +\frac{T^2}{N}$ for any non-preprocessing DDH algorithm (the bound for sqDDH follows by a similar proof). 

For the sqDDH problem, the theorem is sharp, as a simple non-adaptive variant of the adaptive algorithm of~\cite{Corrigan-GibbsK18} 
sketched in Appendix~\ref{app:sqddh}
attains its success probability bound. %(see Appendix~\ref{app:sqddh} for a sketch of this algorithm).
 
For the DDH problem, 
we conjecture that the bound on the success probability is not sharp, and the `right' bound is $\frac{1}{2} + \tilde{O}(\frac{T^2}{N}+\frac{ST}{N})$.
Possibly, the techniques used in the recent result~\cite{AkshimaBGXY24} which determined the maximal success rate of adaptive algorithms can be combined with our techniques to show this improved bound in the non-adaptive setting.

\paragraph{Sharp lower bound for non-adaptive key-recovery attacks on Even-Mansour with preprocessing.} Our last result shows that unbounded preprocessing does not allow speeding up non-adaptive key-recovery attacks on EM (unless $S \geq \tilde{\Omega}(\sqrt{N})$), in stark contrast with adaptive attacks. In the theorem, we denote by $T=T_1+T_2$ the 
total number of queries the online phase of the algorithm makes.
\begin{theorem}\label{thm:em-nonadaptive-intro}
Let $A=(A_0,A_1)$ be a~key-recovery, non-adaptive $(S,T)$-adversary for the Even-Mansour cryptosystem, which can query only the public and encryption oracles and not the decryption oracle. Denote by $\MAXS_{\textup{EM}}(T)$ the optimal success probability of a non-preprocessing, non-adaptive algorithm that makes at most $T$ queries. Then, the success probability of $A$ is at most
$$
2\cdot \MAXS_{\textup{EM}}(T) + \frac{4\log_e(2) S (T+1)}{N} + \frac{T^2}{N} \leq
    \frac{3T^2 }{N}  + \frac{4\log_e(2) S(T+1)}{N}.
$$
Moreover, the theorem also holds for the single-key variant where $k_1=k_2$.
\end{theorem}
The inequality is by \cite{DKS15,EvenM97}.
We remark that by symmetry of the EM construction, we may allow querying the decryption (and not the encryption) oracle. This restriction is also natural in practice, since access to encryption and access to decryption typically arise from different mechanisms, and an adversary may well obtain only one of them.
The theorem is tight by the non-adaptive attacks of~\cite{BW00,Daemen91,DKS15} that can be viewed as preprocessing attacks in the setting of the theorem.\footnote{Similarly to baby-step giant-step for DLOG, these algorithms have a challenge-independent phase that can be performed during preprocessing, and whose outcome can be viewed as an advice string.}

\subsection{Our techniques}
\label{sec:techniques}

Our bounds are obtained using techniques from information theory, as summarized below.

\paragraph{Entropy and Shearer's lemma.} 
For a finite set $\mathcal{X}$, the entropy of a random variable $X$ assuming values in $\mathcal{X}$ is $H(X)=\sum_{x \in \mathcal{X}} \Pr[X=x] \log_e (1/\Pr[X=x])$. 
The entropy, which measures the amount of uncertainty associated with the possible outcomes of $X$, is one of the central notions in information theory.

A basic property of entropy is \emph{subadditivity}: For any random variables $X_1,X_2,\ldots,X_N$, the entropy of the Cartesian product $(X_1,X_2,\ldots,X_N)$ (which assumes values in $\mathcal{X} \times \ldots \times \mathcal{X}$) satisfies the relation $H(X_1,\ldots,X_N)\leq \sum_{i=1}^N H(X_i)$. Shearer's lemma~\cite{ChungGFS86} is a classical inequality that generalizes this property to relate the entropy of a set of random variables to the entropies of collections of their subsets.
\begin{theorem}[Shearer's lemma]
Let $X_1,\ldots,X_N$ be random variables, and let $\Qcal_1,\Qcal_2,\ldots,\Qcal_m$ be subsets of $[N]=\{1,2,\ldots,N\}$, such that each $i \in [N]$ belongs to at least $k$ of them. Then
\[
H(X_1,\ldots,X_N) \leq \frac{1}{k} \sum_{i=1}^m H(X_{\Qcal_i}),
\]
where for $\Qcal =\{i_1,\dots i_\ell\} \subset [N]$, $H(X_{\Qcal}):=H(X_{i_1},X_{i_2},\ldots,X_{i_{\ell}})$.
\end{theorem}

Shearer's lemma was used to obtain numerous results in combinatorics, theoretical computer science, probability theory, and other areas
(see the survey~\cite{Galvin14}).

\paragraph{KL-divergence.} In our applications, it will be convenient to use Shearer's lemma via its version for the \emph{Kullback-Leibler (KL) divergence}, which measures the amount of dissimilarity between two distributions. 
For two distributions $P,Q$ such that the support of $P$ is contained in the support of $Q$, the KL divergence between $P$ and $Q$ is
\[
\KL( P \| Q ) = \sum_{x \in \mathrm{Support}(P)} P(x) \log \frac{P(x)}{Q(x)}. 
\]
The version of Shearer's lemma for KL-divergence reads as follows:
\begin{theorem}[Shearer's lemma for KL-divergence~\cite{GavinskyLSS15}] 
\label{thm:shearer-KL}
For a finite set $\mathcal{X}$, let $Q=(Q_1,Q_2,\ldots,Q_N)$ be the uniform distribution on $\mathcal{X}^N = \mathcal{X} \times \ldots \times \mathcal{X}$, and let $P=(P_1,\ldots,P_N)$ be another distribution on $\mathcal{X}^N$. Let
$\Qcal_1,\Qcal_2,\ldots,\Qcal_m$ be subsets of $[N]$, such that each $i \in [N]$ belongs to at most $k$ of them. Then
$$k\cdot\KL( P \| Q ) \geq \sum_{j \in [m]} \KL( P_{\Qcal_j} \| Q_{\Qcal_j}),$$
where for $\Qcal =\{i_1,\dots i_\ell\} \subset [N]$, $P_{\Qcal}=(P_{i_1},\ldots,P_{i_\ell})$ is the marginal distribution of $P$ on $\mathcal{X}^{\ell}$, and analogously for $Q$. 
\end{theorem}

\paragraph{The concentration inequality of~\cite{GavinskyLSS15} for families of read-$k$ functions.} One of the directions in which Shearer's lemma was applied is showing that functions on a product space which depend on `almost-disjoint' sets of variables behave, in several aspects, similarly to independent functions. A formal manifestation of this phenomenon is the following. 
\begin{definition}
    A family $\{f_1,f_2,\ldots,f_m\}$ of functions over a product space $\mathcal{X}^N$ is called a \emph{read-$k$ family} if for each $1 \leq i \leq N$, at most $k$ of the functions depend on the $i$'th coordinate. 
\end{definition}
An easy application of Shearer's lemma allows showing that if $f_1,\ldots,f_m:\{0,1\}^n \to \{0,1\}$ is a read-$k$ family with $\Pr[f_i(x)=1]=p$ for all $i$, then $\Pr[f_1(x)=f_2(x)=\ldots=f_m(x)=1] \leq p^{m/k}$. That is, the probability that all functions are equal to 1 simultaneously is not much larger than $p^m$ (which would be the probability if the functions were completely independent). 

Gavinsky, Lovett, Saks and Srinivasan~\cite{GavinskyLSS15} used the KL-divergence variant of Shearer's lemma to prove a significantly stronger result in this direction -- a variant of Chernoff's bound for read-once families of Boolean functions. Recall that for independent random variables $X_1,X_2,\ldots,X_m \in \{0,1\}$ such that $\Pr[X_i=1]=p$, Chernoff's inequality asserts $\Pr[\sum_{i=1}^m X_i \geq (p+\epsilon)m] \leq \exp(-2\epsilon^2 m)$. The authors of~\cite{GavinskyLSS15} proved the following.
\begin{theorem}[The concentration bound of~\cite{GavinskyLSS15} for read-$k$ families of Boolean functions]
\label{Thm:Gavinsky-intro}
    Let $f_1,\ldots,f_m:\{0,1\}^N \to \{0,1\}$ be a read-$k$ family of functions. If $\Pr[f_i(x)=1]=p$ for all $i$, then $\Pr[\sum_{i=1}^m f_i(x) \geq (p+\epsilon)m] \leq \exp(-2\epsilon^2 m/k)$.
\end{theorem}
In the 10 years since its introduction, this concentration inequality 
has become very useful and 
was applied to obtain various results in TCS and combinatorics (see, e.g.,~\cite{HatamiHTT21,HiraharaS24,Kahn22,KuniskySWY25}). 

\paragraph{Variants of Shearer's lemma and the concentration inequality of~\cite{GavinskyLSS15} for distributions defined over permutations.} The main technical tools we use in this paper are variants of the above results for distributions defined over the space $S_N$ of permutations. It will be convenient for us to state the results for the space of bijections from $[N]$ to a set with $N$ elements, which is obviously equivalent.

The following variant of Shearer's lemma for this non-product setting is a special case of a result that was proved (in equivalent forms) by Barthe, Cordero-Erausquin, Ledoux, and Maurey~\cite[Proposition~21]{BartheCLM11} and by Caputo and Salez~\cite[Theorem~4]{CaputoS24}. 
\begin{theorem}[Variant of Shearer's inequality for random bijections] 
\label{thm:shearer-perm-intro}
Let $\mathcal{X}$ be a set of size $N$. Let $Q_X = Q_{X_1,\ldots,X_N}$ be the uniform distribution over bijections from $[N]$ to $\mathcal{X}$,
and let $P_X = P_{X_1,\ldots,X_N}$ be another distribution over such bijections. Let 
$\Qcal_1,\Qcal_2,\ldots,\Qcal_m$ be subsets of $[N]$, such that each $i \in [N]$ belongs to at most $k$ of them. Then
$$
2k \cdot  \KL( P_X \| Q_X ) \geq \sum_{j \in [m]} \KL( P_{X_{\mathcal{U}_j}} \| Q_{X_{\mathcal{U}_j}} ),$$
where $P_{X_{\mathcal{U}}}$ is the distribution of $X_{\mathcal{U}} := \left(X_i \mid i \in \mathcal{U} \right)$ with respect to $P$ (and analogously for $Q$).
\end{theorem}
The derivation of Theorem~\ref{thm:shearer-perm-intro} from~\cite[Proposition~21]{BartheCLM11} and~\cite[Theorem~4]{CaputoS24} is presented in Appendix~\ref{sec:sub:Deduction}. Since the proofs of these two results 
%Since the proofs of Theorem~\ref{thm:shearer-perm-intro} in~\cite{BartheCLM11,CaputoS24} 
are somewhat involved, we also present an elementary proof of the theorem, albeit with a worse constant of $9$, in Appendix~\ref{sec:sub:weaker-Shearer}. We note that the constant $2$ in the theorem cannot be improved to $1$, as there are examples in which the ratio between the two sides of the inequality is $\frac{N}{N-1}$. It is not known whether the constant $2$ is optimal in general. 

\medskip

We derive from Theorem~\ref{thm:shearer-perm-intro} the following variant of the concentration inequality of Gavinsky, Lovett, Saks and Srinivasan for $k$-read families~\cite{GavinskyLSS15}, for the setting of bijections.

\begin{theorem}[Concentration for read-$k$ families on bijections] 
\label{thm:Gavinsky-perm}
Let $\mathcal{X}$ be a set of size $N$. Let $Q_X = Q_{X_1,\ldots,X_N}$ be the uniform distribution over bijections from $[N]$ to $\mathcal{X}$,
and let $P_X = P_{X_1,\ldots,X_N}$ be another distribution over such bijections. 
Let $\{ f_j \}_{j \in [m]}$ be a read-$k$ family of functions, with $f_j:\mathcal{X}^N \mapsto [0,1]$ for all $j$. Denote $p_j = \E_{P_X}[f_j(X)]$ and let $p = \frac{1}{m} \cdot \sum_{j \in [m]} p_j$ be the average of the expectations. Similarly, denote $q_j = \E_{Q_X}[f_j(X)]$ and $q = \frac{1}{m} \cdot \sum_{j \in [m]} q_j$.
Then
$$2k \cdot  \KL( P_X \| Q_X ) \geq m \cdot \KL( p \| q ),$$
where $\KL( p \| q ) = p \log(\frac{p}{q}) + (1-p)\log(\frac{1-p}{1-q})$ is the KL-divergence between two Bernoulli distributions with parameters $p$ and $q$.
\end{theorem}
The proof of Theorem~\ref{thm:Gavinsky-perm} is presented in Section~\ref{sec:preliminaries}.
We remark that this theorem easily yields a
concentration result in the form of Theorem~\ref{Thm:Gavinsky-intro} for read-$k$ families on bijections, but we use the above variant as it is more useful to us. 

%We expect that, like Shearer's lemma and the concentration inequality of Gavinsky, Lovett, Saks and Srinivasan~\cite{GavinskyLSS15} for Boolean functions, the new variants will be useful in various other settings in theoretical computer science and combinatorics. % \avichai{What's the purpose of this paragraph?}

\subsection{Proof overview}\label{sec:proof-overview}
We begin by defining a model, referred to as \emph{permutation challenge} (PC), for the problems analyzed in this paper.
This model is inspired by 
the pre-sampling technique~\cite{BartusekMZ19,CorettiDG18,CorettiDGS18,DodisGK17,Unruh07} which allows bounding the success probability of preprocessing adversaries by analyzing non-preprocessing adversaries. 
Analogously, our model allows bounding the success probability of non-adaptive preprocessing adversaries (for a certain class of problems) by analyzing non-preprocessing adversaries.

While the pre-sampling technique is specialized to adaptive adversaries,
our model allows optimizations for non-adaptive adversaries, and uses a completely different set of analytic techniques. The model abstracts away the internal details of the problem that are less relevant, and allows focusing on its most important aspects that allow proving time-space tradeoffs for non-adaptive adversaries. 
The fact that a single model can capture  the wide array of search and decision problems we consider and allows analyzing them collectively is highly non-trivial. 

\paragraph{The permutation challenge model.} At a high-level, in a game in the PC model, there is an ``inner permutation'' 
$\sigma:[N] \mapsto [N]$ and a secret $\secv$ selected from some space.
An adversary $A_1$ is given direct oracle access to $\sigma$ by issuing inner queries. 

In addition, the secret $\secv$ is used to implement an ``outer function'' which wraps the inner permutation. The adversary is allowed to query the outer function, where each query is translated using $\secv$ into a query to $\sigma$ via a \emph{translation function}, and the output of $\sigma$ is post-processed (again using $\secv$) and given back to the adversary as the oracle answer to the outer query.

The goal of $A_1$ is to ``unwrap'' the inner permutation (for example, by recovering the secret or a part of it) after interacting with the inner permutation via inner queries and with the outer function via outer queries. 

In the preprocessing setting, $A_1$ is given an additional advice string $z = A_0(\sigma)$. The game is essentially the same, and we allow the queries of $A_1$ to depend on $z$. We say that an algorithm is non-adaptive if its inner and outer queries (before translation) depend only on $z$ 
(but not on evaluations of $\sigma$ on secret-dependent values).

\paragraph{Instantiations.}
The DLOG problem can be easily reduced to a setting where the encoding $\sigma$ is a permutation on $[N]$, and the secret is the secret discrete-log $d$ which the adversary has to compute. A query of the adversary is a pair of group elements $(a,b)$, which is mapped to the group element $a \cdot d + b \bmod N$ via a linear function applied to the discrete log. 
If $a \equiv 0 \pmod N$, then the query directly accesses the inner permutation, and otherwise, it is an outer query. Here, there is no post-processing of the answer of $\sigma$.

Our model further supports the DDH and sqDDH problems (although their definitions are more technical).

For the Even-Mansour construction, the inner permutation $\sigma$ is the public permutation and the adversary is allowed to query it directly. 
The secret is the key $(k_1,k_2)$, and 
the outer function corresponds to the encryption/decryption oracle. The translation function maps an outer query by XORing it with one of the keys, while the post-processing function XORs the other key to the output of $\sigma$.

\paragraph{The proof.}
The most important component of the model is the translation function that maps each outer query to an input of $\sigma$ using the secret $\secv$.

A translation function is called \emph{uniform} if, for every fixed outer query, the translated input to $\sigma$ is distributed roughly uniformly when the secret $\secv$ is chosen uniformly. We show that if the translation function is uniform, then a non-adaptive preprocessing adversary obtains only a limited advantage (as a function of $S$, $T$, and a uniformity parameter) over a non-preprocessing adversary.

To explain the basic argument, consider first an online non-adaptive adversary $A_1$ that makes no inner queries, and fix a preprocessing string $z$, which fixes all queries of the adversary. We compare two distributions on the inner permutation $\sigma$: the distribution $Q$, where $\sigma$ is uniform, corresponding to the non-preprocessing setting, and the distribution $P$, where $\sigma$ is conditioned on the event $A_0(\sigma)=z$, corresponding to the preprocessing setting. For each secret value $j$, let $f_j(\sigma)$ indicate whether the adversary succeeds when the secret is $j$. The uniformity of the translation function implies that the family $\{f_j\}$ is read-$k$ for some small $k$ (which depends on the number of outer queries of $A_1$ and the uniformity parameter). Therefore, Theorem~\ref{thm:Gavinsky-perm} bounds the gap between the average success probability under $P$ and under $Q$ in terms of $\KL(P\|Q)$. Intuitively, the theorem converts the read-$k$ property of $\{f_j\}$ into a bound showing that the average success probability under $P$ cannot exceed that under $Q$ by much unless $\KL(P\|Q)$ is large. Since the only difference between $P$ and $Q$ is that in $P$ we condition on the short advice string $z$, this  establishes Lemma~\ref{lem:non-fixed g-nonadaptive}, which restricts the adversary to making only outer queries. A more detailed overview of this argument appears in Section~\ref{subsec:proof overview of weaker}.

Lemma~\ref{lem:non-fixed g-nonadaptive} already implies bounds for general adversaries by including the answers to the inner queries in the preprocessing string (Theorem~\ref{thm:g-nonadaptive-weaker}). However, as demonstrated in Remark~\ref{rem:comparison}, this approach is highly suboptimal for decision problems such as DDH and sqDDH. Therefore, we analyze the inner queries directly and more carefully in order to obtain improved bounds (Theorem~\ref{thm:g-nonadaptive}). When handling adversaries that make inner queries, a direct application of Shearer-like inequalities fails, since some permutation indices may always be queried. Nevertheless, we show that Theorem~\ref{thm:Gavinsky-perm} can still be applied to the part of the permutation that has not been directly queried. The proof proceeds via a hybrid argument, based on an intermediate game between the preprocessing and non-preprocessing games, which isolates the effect of the inner queries. The reader is referred to Section~\ref{subsec:weaker bound} for a complete comparison between these proof strategies and the corresponding bounds.
\paragraph{Comparison with other techniques.} 
At a high level, there are currently three main generic techniques for proving cryptanalytic time-space tradeoffs:
(1) compression arguments~\cite{Corrigan-GibbsK18,FreitagGK22,GennaroT00,GolovnevGHPV20,Wee05}, (2)
the pre-sampling technique~\cite{BartusekMZ19,CorettiDG18,CorettiDGS18,DodisGK17,Unruh07}, and (3) concentration inequalities~\cite{AkshimaBGXY24,AkshimaCDW20,AkshimaGL24,ChawinHM20,ChungGLQ20,GuoLLZ21,ImpagliazzoK10}.  

Most of these generic techniques were devised for adaptive algorithms, and it is not clear how to optimize them for non-adaptive algorithms for the problems we consider. 
In Appendix~\ref{app:limitations}, we discuss the limitations of previous techniques in detail and argue that they cannot meaningfully distinguish between adaptive and non-adaptive algorithms in our setting: even when applied to non-adaptive adversaries, they can improve over the tight bounds known for adaptive adversaries by at most a polylogarithmic factor in $N$.

In addition to generic techniques, several papers developed specialized methods for proving time-space lower bounds for specific problems (some of which extended the generic techniques and combined additional methods). 
Examples of such papers include~\cite{ChawinHM20,GajulapalliGK24,GolovnevGPS23}, which deal with non-adaptive algorithms for function inversion, as well as~\cite{ChungL23} which analyzes the 3SUM-Indexing problem.
These specialized techniques seem inapplicable in our setting.

\paragraph{Open problems.} 
The independence of the queries from the challenge appears inherent to our proof strategy. 
Known Shearer-like bounds control information of a random function $f$ only on sets of input–output pairs that are \emph{fixed} and independent of $f$, namely, of the form $\{(a,f(a)) : a\in A\}$ for fixed $A$. Allowing queries to depend on the challenge $c=f(d)$ leads to composed inputs such as $f(d+f(d))$, which no longer have this fixed-set structure. The paper~\cite{RaoS18} derived a Shearer-type inequality that allows for some dependence, but it is too weak for our setting. 
Overall, extending our approach beyond this restriction would require new Shearer-type inequalities or alternative techniques that can handle such challenge-dependent queries. 

More generally, a natural goal to pursue, in view of our results, is to obtain time-space lower bounds for adaptive attacks with preprocessing, as a function of \emph{the number of adaptivity rounds} for the problems considered in this paper. %\avichai{Maybe remove this sentence that became redundant:} The most basic setting is to allow the queries of the adversary to depend on the challenge.
%, corresponding to 2 adaptivity rounds (see Appendix~\ref{app:adaptivity}). 
For the DLOG problem, we conjecture that the success probability of an $(S,T)$-algorithm with $r$ rounds of adaptivity is at most $\tilde{O}(T^2/N + rST/N)$. This matches our result for non-adaptive algorithms (i.e., $1$ round of adaptivity), as well as the bounds of~\cite{Corrigan-GibbsK18} for adaptive algorithms (i.e., $T$ rounds of adaptivity). Furthermore, it would be sharp, as for any $1 \leq r \leq T$, it is matched by a variant of the adaptive algorithm of~\cite{BernsteinL13,Corrigan-GibbsK18,Mih10}, in which instead of constructing one chain of length $T$ one constructs multiple chains of length $r$. 

Another open problem, mentioned above, is to close the gap between upper and lower bounds for non-adaptive $(S,T)$-algorithms for the DDH problem. As written, we conjecture that the bound of Theorem~\ref{thm:DDH-intro1} is not tight, and the optimal bound is  
 $\frac{1}{2} + \tilde{O}(\frac{T^2}{N}+\frac{ST}{N})$.

\paragraph{Organization of the paper.} In Section~\ref{sec:preliminaries} we present in detail the definitions and results related to entropy and KL-divergence that we will use in our proofs, including the variants of Shearer's inequality and of the inequality of Gavinsky, Lovett, Saks and Srinivasan for functions over permutations. In Section~\ref{sec:PC} we present the permutation challenge (PC) game model and show how  DLOG, DDH, sqDDH, and EM key-recovery fit into it. In Section~\ref{sec:main} we prove our main theorem for the PC model.
In Section~\ref{sec:DLOG} we present the bounds for the DLOG, DDH, and sqDDH problems, and in Section~\ref{sec:EM} we present the bound for attacks on the EM cryptosystem. 

\section{Information Theory}
\label{sec:preliminaries}

In this section we present definitions and results from information theory that will be used throughout the paper. For more background on information theory, see~\cite{CoverT06}.

\subsection{Definitions and notations}

Let $\mathcal{X},\mathcal{Y},\mathcal{Z}$ be finite sets. Let $P_{X,Y,Z}$ and $Q_{X,Y,Z}$ be probability distributions over $\mathcal{X} \times \mathcal{Y} \times \mathcal{Z}$. Denote the projections of $P_{X,Y,Z}$ on $\mathcal{X} \times \mathcal{Y}, \mathcal{X},$ and $\mathcal{Z}$, by $P_{X,Y},P_X$, and $P_Z$, respectively, and use similar notations for $Q_{X,Y,Z}$. A random variable $X$ assuming values in $\mathcal{X}$ is said to be drawn from the distribution $P_X$ if $\Pr[X=x]=P_X(x)$ for all $x \in \mathcal{X}$.  

\medskip

\noindent (a). The \emph{entropy} of a random variable $X$ drawn from $P_X$ 
is $$\HH(X) = \E_{P_X(x)}[\log(1/P_X(x))] = \sum_{x \in \mathcal{X}} P_X(x) \log(1/P_X(x) ).$$
Here and throughout the paper, all logarithms are in base $e$, unless explicitly stated otherwise. In case of ambiguity about the distribution, we may also write $\HH(P_X)$.

\medskip

\noindent (b). The \emph{conditional entropy} of $X$ given $Y$ (drawn from $P_{Y}$) is
$$\HH(X \mid Y ) = \E_{P_Y(y)}[\HH(X \mid Y=y)] =  \sum_{y \in \mathcal{Y}} P_Y(y)\HH(X \mid Y=y) = \HH(X,Y) - \HH(Y).$$

\medskip

\noindent (c). The \emph{Kullback-Leibler divergence} (KL-divergence) between two distributions $P_X,Q_X$ is 
$$\KL( P_X \| Q_X ) = \E_{P_X(x)}[\log(P_X(x)/ Q_X(x))] = \E_{P_X(x)}[\log(1/ Q_X(x))] - \HH(P_X),$$
where we assume that the support of $P_X$ is contained in the support of $Q_X$ (otherwise, the KL-divergence is infinite).

\medskip 

\noindent (d). The KL-divergence between $P_X,Q_X$ conditioned on $P_Z$ is
$$\KL( P_{X \mid Z} \| Q_{X \mid Z} ) = \E_{P_Z(z)}[\KL(P_{X \mid Z = z} \| Q_{X \mid Z = z} )].$$

\medskip

\noindent (e). The KL-divergence between two Bernoulli distributions $P_X,Q_X$ with parameters $p,q$, respectively (i.e., $\mathcal{X}=\{0,1\}$, $P_X(1)=p,P_X(0)=1-p$, and similarly for $Q_X$) is denoted by $$\KL(p \| q) = \KL( P_X \| Q_X).$$

\medskip 

\noindent (f). The \emph{mutual information} between $X$ and $Y$ (drawn from $P_{X},P_{Y}$) is
$$\I(X;Y) = \KL(P_{X,Y} \| P_X P_Y).$$

\subsection{Basic properties} 

We shall use the following basic properties of entropy, KL-divergence and mutual information.

\begin{enumerate}[label=(\arabic*), ref=\text{Property (}\arabic*\text{)}]
  \item \label{prope:conditioned entropy} Conditioning does not increase entropy, namely
$$\HH(X) \geq \HH(X \mid Y ),$$
   with equality if and only if $X,Y$ are independent, i.e., $P_{X,Y} = P_X \times P_Y$.
  \item \label{prope:chain rule entropy} The \emph{chain rule for entropy} asserts that $$\HH(X,Y) = \HH(X) + \HH(Y \mid X).$$
Thus, $\HH(X,Y) = \HH(X) + \HH(Y)$ if and only if $X,Y$ are independent.
  \item \label{prope:entropy bound} The entropy of $X$ is upper-bounded by the logarithm of the size of its support, namely
  $$\HH(X) \leq \log|\Supp(\mathcal{X})|.$$
  \item \label{prope:kl nonnegative} KL-divergence is non-negative, namely $$\KL(P_{X} \| Q_{X}) \geq 0.$$
  \item \label{prope:kl convex} KL-divergence is convex, namely for $0 \leq \lambda \le 1$,
  $$\KL( \lambda P_X + (1 - \lambda) P'_X \| \lambda Q_X + (1 - \lambda)Q'_X ) \leq
  \lambda \KL( P_X \| Q_X ) + (1 - \lambda) \KL( P'_X \| Q'_X ).$$
  \item \label{prope:kl and entropy} If $Q_X$ is uniform  over its support (which includes the support of $P_X$), then
    \begin{align*}
    \KL( P_X \| Q_X ) = &\; \E_{P_X(x)}[\log(1/ Q_X(x))] - \HH(P_X) = \E_{Q_X(x)}[\log(1/ Q_X(x))] - \HH(P_X) \\
    = &\; 
    \HH(Q_X) - \HH(P_X).
    \end{align*}
  \item \label{prope:chain rule kl} The \emph{chain rule for KL-divergence} asserts that
  $$\KL( P_{X,Y} \| Q_{X,Y}) =  \KL( P_X \| Q_X ) + \KL( P_{Y \mid X} \| Q_{Y \mid X} ) .$$
  \item \label{prope:data processing}
  The \emph{data processing inequality} asserts that for a function $f: \mathcal{X} \mapsto \mathcal{Y}$,
  $$\KL( P_X \| Q_X) \geq \KL( P_{f(X)} \| Q_{f(X)} ).$$
\item \label{prope:pinsker}
A special case of \emph{Pinsker's inequality} states that
  $$2 (p - q)^2 \leq \KL(p \| q ).$$
  \item \label{prope:information bound} The mutual information satisfies
  \begin{align*}
  &\; \I(X ; Y) = \KL(P_{X,Y} \| P_X P_Y) = 
  \KL(P_{X \mid Y} \| P_X) = \HH(X) - \HH(X \mid Y) \leq 
  \HH(X).
  \end{align*}
\end{enumerate}

\subsection{Lemmas}

We shall use the following lemmas.

\begin{proposition} \label{prop:KL-analysis}
Let $p, \epsilon > 0$ such that $p + \epsilon \leq 1$. 
Then, $$\KL(p + \epsilon \| p) \geq \frac{\epsilon^2}{2(p + \epsilon)}.$$
\end{proposition}

\begin{proof}
The proof follows a standard analytic approach for proving inequalities on KL-divergence. 

By \ref{prope:pinsker}, $\KL(p + \epsilon \| p) \geq 2\epsilon^2$, 
implying the statement whenever $p+\epsilon \ge \frac{1}{4}$. Therefore, we may assume $p+\epsilon \le \frac{1}{4}$. Let us define a function $f:[0,\epsilon] \to \mathbb{R}$ by
\[
    f(x) = \KL(p + x \| p) - \frac{1}{2(p+\epsilon)}x^2.
\]

Since $f(0)=0$ and the statement is equivalent to $f(\epsilon) \ge 0$, it suffices to show that $f$ is monotonically increasing. We have
\[
    f'(x) = \log\left(\frac{p+x}{p}\right)-\log\left(\frac{1-(p+x)}{1-p}\right)-\frac{1}{p+\epsilon}x.
\]

Again, we may notice that $f'(0) = 0$. Therefore, in order to show that $f'(x) \ge 0$ for all $x \in [0,\epsilon]$ and finish the proof, it suffices to show that $f''(x) \ge 0$ for all $x \in [0,\epsilon]$. Since $p+\epsilon \le \frac{1}{4}$, we obtain that for all $x \in [0,\epsilon]$ we have
\[
    f''(x) = \frac{1}{(1-(p+x))(p+x)}-\frac{1}{p+\epsilon} \ge \frac{1}{(1-(p+\epsilon))(p+\epsilon)}-\frac{1}{p+\epsilon}\ge 0,
\]
completing the proof.
\end{proof}

\begin{corollary}
\label{cor:KL-analysis}
Let $0 < p, q \leq 1$. 
Then, $$p \leq 2(q +\KL(p \| q)).$$
\end{corollary}

\begin{proof}
If $p - q < q$, 
then $p < 2 q \leq 2(q + \KL(p \| q))$,
as the KL-divergence is non-negative. Otherwise, $p - q \ge q$.
Applying Proposition~\ref{prop:KL-analysis}, we obtain 
$$
\KL(p \| q) \geq \frac{(p - q)^2}{2q} \geq
\frac{(p - q)}{2}.
$$
Therefore, $p \leq q + 2\KL(p \| q)) \leq 2(q + \KL(p \| q))$.

\end{proof}

The following proposition can be viewed as a form of the data-processing inequality.

\begin{proposition}[\cite{GavinskyLSS15}, Claim 2.6] \label{prop:dpKL}
Let $Q_X$ be the uniform distribution over a set $\mathcal{X}$, let $P_X$ be a distribution over $\mathcal{X}$, and let $f:\mathcal{X} \mapsto [0,1]$. 
Denote $q = \E_{Q}[f(X)]$ and $p = \E_{P}[f(X)]$.
Then
$$\KL( P_X \| Q_X ) \geq \KL(p \| q ).$$
\end{proposition}

\subsection{A variant of the inequality of~\cite{GavinskyLSS15} for functions over bijections}

In this subsection we derive Theorem~\ref{thm:Gavinsky-perm}, namely, the variant of the concentration bound of Gavinsky, Lovett, Saks and Srinivasan~\cite{GavinskyLSS15} for read-$k$ families of functions over bijections, from the variant of Shearer's lemma for the same setting (Theorem~\ref{thm:shearer-perm-intro} above). Let us restate these two theorems:

\medskip \noindent \textbf{Theorem~\ref{thm:shearer-perm-intro}.} Let $\mathcal{X}$ be a set of size $N$. Let $Q_X = Q_{X_1,\ldots,X_N}$ be the uniform distribution over bijections from $[N]$ to $\mathcal{X}$,
and let $P_X = P_{X_1,\ldots,X_N}$ be another distribution over such bijections. Let 
$\Qcal_1,\Qcal_2,\ldots,\Qcal_m$ be subsets of $[N]$, such that each $i \in [N]$ belongs to at most $k$ of them. Then
$$
2k \cdot  \KL( P_X \| Q_X ) \geq \sum_{j \in [m]} \KL( P_{X_{\mathcal{U}_j}} \| Q_{X_{\mathcal{U}_j}} ),$$
where $P_{X_{\mathcal{U}}}$ is the distribution of the vector $X_{\mathcal{U}} := \left(X_i \mid i \in \mathcal{U} \right)$ with respect to $P$ (and analogously for $Q$).

\medskip \noindent \textbf{Theorem~\ref{thm:Gavinsky-perm}}.
Let $\mathcal{X}$ be a set of size $N$. Let $Q_X = Q_{X_1,\ldots,X_N}$ be the uniform distribution over bijections from $[N]$ to $\mathcal{X}$,
and let $P_X = P_{X_1,\ldots,X_N}$ be another distribution over such bijections. 
Let $\{ f_j \}_{j \in [m]}$ be a read-$k$ family of functions, with $f_j:\mathcal{X}^N \mapsto [0,1]$ for all $j$. Denote $p_j = \E_{P_X}[f_j(X)]$ and let $p = \frac{1}{m} \cdot \sum_{j \in [m]} p_j$ be the average of the expectations. Similarly, denote $q_j = \E_{Q_X}[f_j(X)]$ and $q = \frac{1}{m} \cdot \sum_{j \in [m]} q_j$.
Then
$$2k \cdot  \KL( P_X \| Q_X ) \geq m \cdot \KL( p \| q ),$$
where $\KL( p \| q ) = p \log(\frac{p}{q}) + (1-p)\log(\frac{1-p}{1-q})$ is the KL-divergence between two Bernoulli distributions with parameters $p$ and $q$.

\medskip

\begin{proof}[of Theorem~\ref{thm:Gavinsky-perm}]
Let $\{\mathcal{U}_j \}_{j \in [m]}$ be a family of index sets such that for each $j \in [m]$, $f_j$ depends only on the coordinates in $\mathcal{U}_j$. Since for each 
$i \in [N]$, $|\{ j \in [m] \mid i \in \mathcal{U}_{j} \} | \leq k$, Theorem~\ref{thm:shearer-perm-intro} implies
$$2k \cdot  \KL( P_X \| Q_X ) \geq \sum_{j \in [m]} \KL( P_{X_{\mathcal{U}_j}} \| Q_{X_{\mathcal{U}_j}} ).$$
For every $j \in [m]$, $Q_{X_{\mathcal{U}_j}}$ is uniform over its support (which includes the support of $P_{X_{\mathcal{U}_j}}$).
Thus, we apply Proposition~\ref{prop:dpKL} with $f_j$ and deduce 
$$\KL( P_{X_{\mathcal{U}_j}} \| Q_{X_{\mathcal{U}_j}} ) \geq \KL( p_j \| q_j ).$$
Hence, 
\begin{align*}
2k \cdot  \KL( P_X \| Q_X ) 
\geq  \sum_{j \in [m]} \KL( p_j \| q_j ) 
= m \sum_{j \in [m]} (1/m) \KL( p_j \| q_j ) 
\geq  m \cdot \KL( p \| q ),
\end{align*}
where the final inequality holds by convexity of the KL-divergence (\ref{prope:kl convex}).
\end{proof}

\begin{remark}
Note that while the definition of a read-$k$ family requires the functions $\{f_j\}$ to be defined over the entire domain 
$\mathcal{X}^N$, only their values on bijections are relevant in Theorem~\ref{thm:Gavinsky-perm}. Consequently, in the applications of this theorem we define each function only over bijections, while implicitly fixing it to $0$ on inputs that are not bijections. 
\end{remark}

\section{The Permutation Challenge Model}
\label{sec:PC}

\subsection{Non-preprocessing setup} 

A permutation challenge (PC) game  
$$\PC := \PC(N,\SECS,\mathcal{M},\TR,\POST,\SUC)$$
is defined as follows. 

For an integer $N > 0$, let $\sigma:[N] \mapsto [N]$ be a permutation and let $\SECS$ be a space of secrets.

Let $\TR:\SECS \times \mathcal{M} \rightarrow [N]$ be a translation function that obtains 
as inputs a secret $\secv \in \SECS$ (e.g., the discrete log in the DLOG problem) and an outer query $m \in \mathcal{M}$ (from the outer query space $\mathcal{M}$) and returns
a translated query to $\sigma$, denoted $\TR(\secv,m) \in [N]$.
Let $\POST: \SECS \times [N] \rightarrow [N]$ be a post-processing function that receives the secret and an inner permutation output and returns a value in $[N]$ (the output space of the outer function).

\begin{definition}
A translation function $\TR$ is called $u$-uniform if for every $m \in \mathcal{M}$ and $j \in [N]$, 
$$ |\{\secv \in \SECS \mid \TR(\secv,m) = j\}| \leq \tfrac{|\SECS|}{u}.$$
\end{definition}
\noindent In several of our applications, the translation function $\TR$ will be $N$-uniform, which is the largest value possible.

Let $A_1$ be an adversary for $\PC$. We assume for simplicity that $A_1$ is deterministic. As noted in Remark~\ref{rem:deterministic} below, this assumption is without loss of generality.

$A_1$ has oracle access to two related oracles. The inner oracle allows $A_1$ to issue an \emph{inner query} $i \in [N]$ to $\sigma$ and obtain $\sigma(i)$.
In some applications, $A_1$ can also issue inner queries $i \in [N]$ to the inverse inner permutation and obtain $\sigma^{-1}(i)$.

The outer oracle allows $A_1$ to issue an \emph{outer query} $m \in \mathcal{M}$ and obtain 
$\POST(\secv,\sigma(\TR(\secv,m)))$.

After the interaction with the oracles 
which we denote in short by $\mathcal{O}(\sigma,\secv)$, $A_1^{\mathcal{O}(\sigma,\secv)}$ outputs a value $v$ 
from some domain $\mathcal{V}$.

Let $\SUC := \SUC_{A_1^{\mathcal{O}(\sigma,\secv)}}(\secv)$ be a $0/1$ success predicate that obtains $\secv$ and has the same oracle accesses as $A_1^{\mathcal{O}(\sigma,\secv)}$. This predicate outputs $1$ if $A_1$ succeeds and $0$ otherwise.

Formally, $\SUC$ has two phases. In the first phase, it simulates $A_1^{\mathcal{O}(\sigma,\secv)}$ and obtains its output. In the second phase, it outputs a value by applying a $0/1$ predicate that receives as input the secret $\secv$, the output of $A$, and all the query answers obtained by $A$.\footnote{In our applications, the predicate only needs the output of $A$, but we allow it to obtain the query answers obtained by $A$ for generality.}
We may also allow $\SUC$ to make additional queries to verify the success of $A$, which will be accounted for in the total time complexity (but we do not use this possibility in our applications).

We remark that when we write
$\PC := \PC(N,\SECS,\mathcal{M},\TR,\POST,\SUC)$,
the success predicate is viewed as an interface that is instantiated given an adversary $A_1$.

Let $P_{\Sigma,\SECV} = P_{\Sigma} P_{\SECV}$ be a distribution such that $P_{\Sigma}$ is uniform over permutations
and $P_\SECV$ is uniform over the secret space $\SECS$.

The success probability of $A_1$ is defined as
$$\E_{P_{\Sigma,\SECV}}[\SUC_{A_1^{ \mathcal{O}(\Sigma,\SECV)} }(\SECV)].$$

We measure the complexity of $A_1$ in terms of its number of queries to the inner permutation, $T_1$, and its number of queries to the outer function, $T_2$, and define $T = T_1 + T_2$.

\paragraph{Non-adaptivity.}
We say that $A_1$ is non-adaptive if its queries are fixed in advance and do not depend on $\sigma$. 

\subsection{Preprocessing setup}
\label{subsec:preprocess setup}

Let $A = (A_0,A_1)$ be a pair of deterministic preprocessing and online algorithms for $\PC$. 
\begin{remark}
\label{rem:deterministic}
In our context, we may assume that $A$ is deterministic without loss of generality, as our lower bound proofs hold for every randomness string shared by $A_0$ and $A_1$.
\end{remark}

The preprocessing algorithm $A_0$ gets direct access to $\sigma$ as input and outputs an advice string $z=A_0(\sigma)$.
The online algorithm $A_1$ takes as an additional input the advice string $z$. We denote the algorithm $A_1$, when it is executed with a preprocessing string $z$, by $(A_1)_{z}$.

We call $(A_0,A_1)$ an $(S,T)$-algorithm if the length of $z = A_0(\sigma)$ is bounded by $S$ bits (for all $\sigma$) and $A_1$ makes at most $T$ oracle queries to both the inner permutation and the outer function. 

We extend the non-preprocessing success predicate $\SUC$ to the preprocessing model as follows: 
it additionally obtains an advice string $z$.
As previously, 
$\SUC_{(A_1)_{z}^{\mathcal{O}(\sigma,\secv)}}(\secv)$ has two phases. In the first phase, it simulates $(A_1)_z^{\mathcal{O}(\sigma,\secv)}$ and obtains its output. 
In the second phase, it outputs a value by applying the same non-preprocessing $0/1$ predicate that receives as input the secret $\secv$, the output of $A_1$ and all the query answers obtained by $A_1$.\footnote{In our applications, the predicate only compares the output of $(A_1)_z$ to some function of the secret.} 

We extend the distribution $P_{\Sigma,\SECV}$ as $P_{\Sigma,Z,\SECV} = P_{\Sigma,Z} \times P_\SECV$,
where 
$P_{\Sigma,Z} = P_{\Sigma} P_{Z \mid \Sigma}$ such that $P_{Z \mid \Sigma = \sigma}(z) = 1$ if 
$z = A_0(\sigma)$.

The success probability $(A_0,A_1)$ in solving the problem is defined as
$$p := p(A_0,A_1) := \E_{P_{\Sigma,Z,\SECV}}[\SUC_{(A_1)_{Z}^{\mathcal{O}(\Sigma,\SECV)}}(\SECV)].$$

\begin{remark}
\label{remark:z}
Crucially, the predicate applied by $\SUC$ after simulating $(A_1)_z^{\mathcal{O}(\sigma,\secv)}$ does not receive $z$ as input.
Namely, $\SUC$ must be able to verify whether the output of $(A_1)_z$ is correct on $(\sigma,\secv)$ \emph{independently} of $z$. Technically, our analysis will sometimes fix $z$ that is not equal to $A_0(\sigma)$, but we require that the success of $(A_1)_z$ will still be defined with respect to $(\sigma,\secv)$.

Another way to state this remark is that it is sufficient to define the predicate $\SUC$ for non-preprocessing algorithms. The extension to the preprocessing model is well-defined and it simply simulates $(A_1)_z$ as a black-box and checks its success with respect to $(\sigma,\secv)$ as in the non-preprocessing model.
\end{remark}

\paragraph{Non-adaptivity.}
We say that $A_1$ is non-adaptive if given any $z$, its queries are fixed and do not (further) depend on $\sigma$. 
Note that if $A_1$ is non-adaptive then
$\SUC := \SUC_{(A_1)_z^{\mathcal{O}(\sigma,\secv)}}(\secv)$ is a non-adaptive predicate. 

We denote by $\mathcal{S} \subseteq [N]$ the set of inner queries of $A_1$ (which may include inverse queries) and by 
$\mathcal{U} \subseteq \mathcal{M}$ the set of its outer queries.
To simplify notation, we assume that $\mathcal{S}$ refers to queries to $\sigma$ and not $\sigma^{-1}$ using the conversion that if $\sigma^{-1}(i) = j$, then $\sigma(j) = i$ is a query to $\sigma$.
Both $\mathcal{S}$ and $\mathcal{U}$ may depend on $z$,
and we sometimes emphasize this by writing (for example) $\mathcal{S}(z)$. 

\subsection{Instantiation in the 
 generic group model}
\label{sec:inst GGM}
We first define the general setting of GGM which is common to all problems in the model.

Let $A_1$ be a generic group algorithm in $\mathbb{Z}_N$ such that $N$ is prime. 
We assume that $A_1$ knows the image of $\sigma:\mathbb{Z}_N \rightarrow \mathcal{W}$, as this knowledge can only increase its success probability. Thus, we restrict $\sigma$ to a subset of $\mathcal{W}$ of size $N$.
Furthermore, by renaming the symbols of the image of $\sigma$, we assume without loss of generality that 
$\sigma:\mathbb{Z}_N \rightarrow \mathbb{Z}_N$ is a uniform bijection from $\mathbb{Z}_N$ to itself.
We represent the elements of $\mathbb{Z}_N$ using $[N]$, and we can thus write $\sigma:[N] \rightarrow [N]$. In particular, note that $N \bmod N = 0$.

We remark that in GGM, we do not allow $A_1$ to query $\sigma^{-1}$.

We now define the DLOG, DDH, and sqDDH problems as permutation challenge games.
\paragraph{Discrete-log.} Let $A_1$ be a discrete-log algorithm for $\mathbb{Z}_N$.
We define a permutation challenge game
$\PC := \PC_{DL}(N,\SECS,\mathcal{M},\TR,\POST,\SUC)$ for the discrete-log problem as follows.

Let $\SECS = [N]$ (identified with $\mathbb{Z}_N$) be the space of secrets of discrete logarithms.
Let $\POST(d,j) = j$ be the trivial post-processing function (as $A_1$ always sees direct outputs of $\sigma$).

Let $\mathcal{M} = [N-1] \times [N]$ be the space of outer queries
and define 
$$\TR(d,(a,b)) = a \cdot d + b \mod N$$ as the translation function that maps a query $(a,b)$
(where $a \neq N$) to a group element according to the linear function $a \cdot d + b \mod N$.
Note that in the representation, an inner query issued to $\sigma$ corresponds to the pair $(N,b)$
(but we view it as directly accessing $\sigma(b)$).

The adversary receives as input the values $\sigma(1)$ and $\sigma(d)$.
However, for simplicity we assume that these are given as outputs of the queries 
$(N,1)$ and $(1,N)$, and the adversary receives no input 
(we make similar assumptions about DDH and sqDDH, defined below). 

Finally, 
$\SUC_{A_1^{\mathcal{O}(\sigma,d)}}(d)$
simply simulates $A_1$ and returns $1$ if it outputs the discrete-log secret~$d$.

\paragraph{DDH.}
Let $A_1$ be a DDH generic algorithm for $\mathbb{Z}_N$.
We define a permutation challenge game
$\PC := \PC_{DDH}(N,\SECS,\mathcal{M},\TR,\POST,\SUC)$ for the DDH problem as follows.

Let $\SECS = [N]^3 \times \{0,1\}$ be the secret space, consisting of $3$ group elements $d_1,d_2,d_3$ and a bit $k$ that $A_1$ needs to output.
Let $\POST(d,j) = j$ be the trivial post-processing function.

Let $\mathcal{M} = [N]^4 \setminus (\{(N,N,N)\} \times [N])$ be the space of outer queries, which consists of $4$ group elements that specify a multi-linear function, denoted by $(a_1,a_2,a_3,b)$,
where we do not allow $a_1 = a_2 = a_3 = N$ (as this corresponds to an inner query). 

Define the translation function
$\TR((d_1,d_2,d_3,k),(a_1,a_2,a_3,b))$ as  
follows:
\begin{align*}
\TR((d_1,d_2,d_3,k),(a_1,a_2,a_3,b) ) =
\begin{cases}
a_1 \cdot d_{1} + a_2 \cdot d_{2} + a_3 \cdot d_{3} + b \mod N & \text{if } k=0, \\
a_1 \cdot d_{1} + a_2 \cdot d_{2} + a_3 \cdot (d_{1} d_{2}) + b \mod N & \text{if } k = 1.
\end{cases}
\end{align*}

Note that if   
$k = 0$, the algorithm effectively receives $\sigma(d_1),\sigma(d_2),\sigma(d_3)$, while if
$k = 1$, the algorithm effectively receives $\sigma(d_1),\sigma(d_2),\sigma(d_1 d_2 \bmod N)$, as in DDH.

As in the discrete-log game, an inner query can be viewed as directly accessing $\sigma$.

Finally, 
$\SUC_{A_1^{\mathcal{O}(\sigma,d_1,d_2,d_3,k)}}(d_1,d_2,d_3,k)$
simply simulates $A_1$ and returns $1$ if it outputs the bit~$k$.

\paragraph{sqDDH.} 
Let $A_1$ be a sqDDH generic algorithm for $\mathbb{Z}_N$.
We define a permutation challenge game
$\PC := \PC_{sqDDH}(N,\SECS,\mathcal{M},\TR,\POST,\SUC)$ for the sqDDH problem as follows.

Let $\SECS = [N]^2 \times \{0,1\}$ be the space of secrets, consisting of $2$ group elements $d_1,d_2$ and a bit $k$ that $A_1$ needs to output.
Let $\POST(d,j) = j$ be the trivial post-processing function.
 
Let $\mathcal{M} = [N]^3 \setminus (\{(N,N\} \times [N])$ be the space of outer queries, which consist of $3$ group elements that specify a multi-linear function, denoted by $(a_1,a_2,b)$,
where we do not allow $a_1 = a_2 = N$ (as this corresponds to an inner query). 

Define the translation function
$\TR((d_1,d_2,k),(a_1,a_2,b))$ as  
follows:
\begin{align*}
\TR((d_1,d_2,k),(a_1,a_2,b) ) =
\begin{cases}
a_1 \cdot d_{1} + a_2 \cdot d_{2}  + b \mod N & \text{if } k=0, \\
a_1 \cdot d_{1} + a_2 \cdot (d_{1})^2 + b \mod N & \text{if } k = 1.
\end{cases}
\end{align*}
Note that if   
$k = 0$, the algorithm effectively receives $\sigma(d_1),\sigma(d_2)$, while if
$k = 1$, the algorithm effectively receives $\sigma(d_1),\sigma((d_1)^2 \bmod N)$, as in sqDDH.

Once again, an inner query is viewed as directly accessing $\sigma$.

Finally, 
$\SUC_{A_1^{\mathcal{O}(\sigma,d_1,d_2,k)}}(d_1,d_2,k)$
simply simulates $A_1$ and returns $1$ if it outputs the bit~$k$.

\subsubsection{Uniformity of the translation function in GGM.}

In all the $\PC$ games in GGM defined above, the $\TR$ function is a multi-variate polynomial of degree $1$ or $2$, in up to $3$ group elements, chosen uniformly at random from $[N]$. The coefficients of this polynomial are determined by the query. For such queries, 
the Schwartz–Zippel lemma~\cite{Schwartz80} immediately gives the following general result.
\begin{lemma}[$u$-uniformity of $\TR$ in the GGM]
\label{lem:ggm uniform}
Let 
$\PC(N,\SECS,\mathcal{M},\TR,\POST,\SUC)$ be a permutation challenge game in the GGM. Assume that $\TR$ is a multi-variate polynomial (with coefficients determined by the query) of degree  $m > 0$ in $v$ variables chosen uniformly at random from $[N]$ with $N$ prime. Then, $\TR$ is $\frac{N}{m}$-uniform.
\end{lemma}

\subsection{Instantiation for the Even-Mansour cryptosystem}
\label{sec:inst EM}

We represent plaintexts, ciphertexts and keys using $[N]$, and we can thus write $\sigma:[N] \rightarrow [N]$.
Using this encoding, the bit representation of any $a \in [N]$ (used when XORing values in the domain) is the standard bit representation of $a - 1$.

We define a permutation challenge game for the key recovery. 
Let $\SECS = [N]^2$ be the space of secrets, consisting of pairs $(k_1,k_2)$ (in case of the single-key scheme, $k_1 = k_2$ and $\SECS = [N]$).
Let $\mathcal{M} = [N]$ be the space of outer queries, which are chosen encryption messages.
For a secret $(k_1,k_2) \in [N]^2$ and
a message (outer query) $m \in [N]$,
define 
$\TR((k_1,k_2),m) = m \oplus k_1$, and let $\POST((k_1,k_2),j) = j \oplus k_2$.

Note that an inner query that is issued to $\sigma$ corresponds to 
a call to the public encryption oracle. Here, we allow the adversary to query $\sigma^{-1}$ as well.

Let $A_1$ be a key-recovery algorithm for EM.
For the permutation challenge game $\PC := \PC_{EM-KR}(N,\SECS,\mathcal{M},\TR,\POST,\SUC)$,
the success predicate
$\SUC_{A_1^{\mathcal{O}(\sigma,(k_1,k_2))}}(k_1,k_2)$
simply simulates $A_1$ and returns $1$ if it outputs the key $(k_1,k_2)$.

\section{Main General Result}
\label{sec:main}

In this section we prove our main result in the permutation challenge game model.

\begin{theorem} \label{thm:g-nonadaptive}
Let $\PC := \PC(N,\SECS,\mathcal{M},\TR,\POST,\SUC)$ be a permutation challenge game, such that 
$\TR$ is $u$-uniform,
and $\SUC$ compares the output of $(A_1)_z$ to some function of the secret.
Let $(A_0,A_1)$ be an $(S,T)$ non-adaptive algorithm with preprocessing for $\PC$.
Denote by $\MAXS(T)$ the optimal success probability (with respect to $\SUC$) of a non-preprocessing, non-adaptive algorithm that makes at most $T$ queries.
Then, the success probability of $(A_0,A_1)$ is at most
\[
\min \left(2\cdot \MAXS(T) + \frac{4\log(2) S T }{u} + \frac{T^2}{u}, \MAXS(T) + 
\sqrt{\frac{\log(2) ST} {{u}}} + \frac{T^2}{2u} \right).
\]
\end{theorem}
We note that a slightly better bound can be achieved for problems with a trivial post-processing function $\POST(\secv,j)=j$. For further details, see Appendix~\ref{app:improvement no post process}.

\subsection{Discussion}

Before proving the theorem, several important remarks are due.
\paragraph{No need to analyze the preprocessing setting.} The power of the theorem comes from the fact that there is no need to analyze the preprocessing setting. Indeed, it suffices 
to instantiate the permutation challenge game 
$\PC := \PC(N,\SECS,\mathcal{M},\TR,\POST,\SUC)$, to
prove that $\TR$ is a $u$-uniform translation function (ideally, for $u \geq \Omega(N)$), and to bound 
$\MAXS(T)$ (i.e., the success probability in the non-preprocessing setting). 
Plugging these values into Theorem~\ref{thm:g-nonadaptive}
immediately bounds the success probability of an adversary in the preprocessing setting.

\paragraph{Generality of $\SUC$.}
We prove the theorem only for a limited class of success predicates 
that compare the output of $(A_1)_z$ to some function of the secret.
Yet, almost all the steps of the proof apply to  arbitrary non-adaptive success predicates (under the restrictions defined in Section~\ref{sec:PC}),
where the only exception that uses this restriction is Lemma~\ref{lem:middle to non preprocess}. 
This lemma deals with a specific game in the non-preprocessing setting,
and it is not possible to prove it in general for all predicates,
as its statement is false for some artificial predicates.

However, for all ``natural'' predicates we are aware of, 
it is easy to extend the theorem tightly by extending the proof 
of Lemma~\ref{lem:middle to non preprocess} accordingly. 
This may require extending $\SUC$ to make additional non-adaptive queries, which are accounted for in the parameter $T$.
For example, it is possible to support selective forgery attacks,
where the goal of the adversary is to predict the value of a predefined outer query 
(that the adversary is not allowed to make). 
This is done by extending $\SUC$ to make this additional (non-adaptive) outer query,
thus adjusting the total query complexity to $T+1$.
Now, $\SUC$ verifies success by comparing the outcome of the query to the output of $(A_1)_z$.
The proof is then adjusted by proving a corresponding variant of Lemma~\ref{lem:middle to non preprocess}.

\paragraph{$\TR$ vs.~$\POST$, and treatment of inner queries to $\sigma^{-1}$.}
The $u$-uniformity of $\TR$ is crucially used in the proof. On the other hand, there are no additional requirements on 
the function $\POST$, and indeed it does not play any direct role in the proof. Its relevance will be in applications of this theorem (in particular, it is important for bounding $\MAXS(T)$).
Similarly, the proof holds regardless of whether $A_1$ can issue inner queries to $\sigma^{-1}$ in addition to $\sigma$, yet this fact may be important for bounding $\MAXS(T)$. 

%\avichai{inner queries to $\sigma^{-1}$ in addition or instead of $\sigma$?}

\subsection{Warm-up: a weaker bound}
\label{subsec:weaker bound}

%\avichai{This is the main section I change, including the title.}

\paragraph{Inner and outer queries.}
The proof of Theorem~\ref{thm:g-nonadaptive} is built from two complementary ideas, which we present separately to keep the exposition modular.
The first idea applies directly only to adversaries that make \emph{no inner queries} (Lemma~\ref{lem:non-fixed g-nonadaptive}); as shown in Theorem~\ref{thm:g-nonadaptive-weaker}, it can nevertheless be applied to obtain a bound for general adversaries, albeit a weaker one.
The second idea establishes a subtle reduction from general adversaries to adversaries without inner queries that preserves the parameters more tightly.
Combining this reduction with the first idea yields the stronger bound of Theorem~\ref{thm:g-nonadaptive}, proved in Section~\ref{subsec:proof of strong version}.

Intuitively, a permutation variant of Shearer’s lemma (Theorem~\ref{thm:shearer-perm-intro}) provides a natural tool for quantifying the (limited) advantage gained from outer queries.
The two arguments therefore share the same Shearer-based core; they differ only in the way they handle the inner queries—the weaker bound neutralizes them in a straightforward manner, while the stronger bound controls them more delicately through the reduction.

\begin{theorem}[Weaker version of Theorem~\ref{thm:g-nonadaptive}] 
\label{thm:g-nonadaptive-weaker}
In the setting of Theorem~\ref{thm:g-nonadaptive}, the success probability of $(A_0,A_1)$ is at most
\[
\min \left(2\cdot \MAXS(T) + \frac{4\log(2) T (S+T \log(N))}{u}, \MAXS(T) + 
\sqrt{\frac{\log(2) T (S+T \log(N))} {{u}}} \right).
\]
\end{theorem}

\begin{remark}  \label{rem:comparison} Comparison between Theorems~\ref{thm:g-nonadaptive} and~\ref{thm:g-nonadaptive-weaker}:
%One could potentially include the answers to the inner queries in the preprocessing string, thus eliminating them. However, this blows up the preprocessing string to length $S + T \log N$ and gives weaker results.
Theorems~\ref{thm:g-nonadaptive} and~\ref{thm:g-nonadaptive-weaker} may look similar, and they indeed obtain comparable bounds for several problems. For example, for search problems (such as in Theorem~\ref{thm:dl-nonadaptive-intro} and Theorem~\ref{thm:em-nonadaptive-intro}), Theorem~\ref{thm:g-nonadaptive-weaker} yields the bound $O(\frac{(S + T \log N) \cdot T}{N} + \frac{T^2}{N})$, instead of the $O(\frac{S \cdot T}{N} + \frac{T^2}{N})$ bound of Theorem~\ref{thm:g-nonadaptive}, which constitutes only a logarithmic penalty. On the other hand, for applications of Theorem~\ref{thm:g-nonadaptive-weaker} to decision problems it yields much weaker results. Specifically, for Theorem~\ref{thm:DDH-intro1}, 
we would obtain 
$$\frac{1}{2} + O\left(\sqrt{\frac{(S + T \log N) \cdot T}{N}}+ \frac{T^2}{N}\right), \qquad \mbox{instead of} \qquad 
\frac{1}{2} + O\left(\sqrt{\frac{S \cdot T}{N}} + \frac{T^2}{N}\right),$$
which is significantly worse for $S \ll T \ll \sqrt{N}$. For example, for $S=O(1)$ and $T = O(N^{1/3})$, the bound of Theorem~\ref{thm:g-nonadaptive-weaker} on the adversary's advantage over a coin toss is larger by a multiplicative factor of $N^{1/6}$. This motivates our more involved proof that handles the inner queries in a novel way.
\end{remark}

Theorem~\ref{thm:g-nonadaptive-weaker} follows easily from the following variant, which forces the adversary to make only outer queries:

\begin{lemma}[Restricted variant of Theorem~\ref{thm:g-nonadaptive}] \label{lem:non-fixed g-nonadaptive}
In the setting of Theorem~\ref{thm:g-nonadaptive}, suppose that $A_1$ makes no inner queries.
In addition, let $\SUC$ be an arbitrary success predicate, as defined in Section~\ref{sec:PC}.
Then, the success probability of $(A_0,A_1)$ is at most
\[
\min \left(2\cdot \MAXS(T) + \frac{4\log(2) S T }{u}, \MAXS(T) + 
\sqrt{\frac{\log(2) ST} {{u}}} \right).
\]
\end{lemma}

\begin{proof}[of Theorem~\ref{thm:g-nonadaptive-weaker} using Lemma~\ref{lem:non-fixed g-nonadaptive}] Since inner queries are independent of the secret and the challenge, they can be incorporated into the preprocessing phase. Specifically, we may assume that $A_0$ makes all inner queries of $A_1$ and adds their outputs to the advice string. Then $A_1$ can simply read these outputs from the advice string instead of issuing the queries itself. This modification increases the length of the advice string to at most $S + T \log(N)$, transforms $A_1$ into a no-inner-queries adversary, and does not affect the success probability. The desired bound then follows directly from Lemma~\ref{lem:non-fixed g-nonadaptive}.

\end{proof}

\subsubsection{Proof overview of Lemma~\ref{lem:non-fixed g-nonadaptive}.}
\label{subsec:proof overview of weaker}

\begin{itemize}
    \item Recall that in the distribution $P_{Z,\Sigma}$, $\Sigma$ is a uniform permutation and $Z$ is determined by $\Sigma$ as 
    $Z=A_0(\Sigma)$.
    The proof essentially shows that a preprocessing adversary whose oracle is distributed as $P_{\Sigma \mid Z}$ cannot do much better that a non-preprocessing adversary whose oracle is distributed as $P_{\Sigma}$. Formally, this is done by defining a distribution $Q$, in which $\Sigma$ and $Z$ are independent, and the success probability of a non-preprocessing algorithm is bounded by $\MAXS(T)$ (Claim~\ref{clm:non-fixed basic}). 
    \item We define $\kappa_z$ to measure the amount of information that $A_1$ obtains if the preprocessing algorithm $A_0$ outputs the value $z$, compared to a non-preprocessing adversary under $Q$ that obtains no information. We show that since $z$ is of length $S$ bits, on average over $z$, $\kappa_z \leq \log(2)S$ (Claim~\ref{clm:non-fixed klz}), confirming the intuitive insight that a string of $S$ bits can provide at most $S$ bits of information on average.
    \item Claim~\ref{clm:non-fixed main-eq2} is the heart of the proof. We fix $Z=z$, and show that the success probabilities of adversaries with and without $z$ are very close (as a function of $\kappa_z$ and $T/u$). This is done by observing that the translated outer queries (that are input to $\Sigma$) are $u$-uniform for a uniformly sampled secret $\secv$. It follows that the probability that $A_1$ queries any fixed $i \in [N]$ of $\Sigma$ is at most $\tfrac{T}{u}$, which allows us to apply variants of Shearer's lemma.
    \item The theorem follows by averaging over Claim~\ref{clm:non-fixed main-eq2} and applying Claim~\ref{clm:non-fixed klz} to bound this average.
\end{itemize}

\subsubsection{Proof of Lemma~\ref{lem:non-fixed g-nonadaptive}.}
We define a distribution $Q_{\Sigma,Z,\SECV}$, in which $\Sigma$ and $Z$ are distributed as in $P$, but are now sampled independently. Specifically, $Q_{\Sigma,Z} = P_\Sigma P_Z$. The conditioned distribution of $\SECV$ remains unchanged, i.e., $Q_{\SECV \mid \Sigma,Z} = P_{\SECV \mid \Sigma,Z}$. Intuitively, the distribution $Q$ represents running the algorithms $(A_0,A_1)$, where each algorithm queries a different bijection, effectively making $A_1$ a non-preprocessing algorithm. Denoting the success probability of $(A_0,A_1)$ with respect to $Q$ by
$q(A) := \E_{Q_{\Sigma,Z,\SECV}}
[\SUC_{(A_1)_{Z}^{\mathcal{O}(\Sigma,\SECV)}}(\SECV)]$, we obtain the following claim as a direct corollary:

\begin{claim} \label{clm:non-fixed basic}
$$
q(A)
\leq 
\MAXS(T).$$
\end{claim}

\noindent Let
$$\kappa_{z} =  \KL( P_{\Sigma \mid Z = z} \| Q_{\Sigma \mid Z = z} )$$
measure the amount of information $A_1$ obtains if the preprocessing algorithm $A_0$ outputs the value $z$. We have:

\begin{claim} \label{clm:non-fixed klz}
\begin{align} 
\E_{P_{Z}(z)}[\kappa_{z}] \leq \log(2) S.
\end{align}
\end{claim}

\begin{proof}
Recall that $Q_{Z,\Sigma} = P_Z P_\Sigma$. Therefore, by \ref{prope:information bound} above, we obtain:
\[
    \E_{P_{Z}(z)}[\kappa_{z}] =: \KL( P_{\Sigma \mid Z} \| Q_{\Sigma \mid Z} ) = \KL( P_{\Sigma \mid Z} \| P_{\Sigma} ) = \KL(P_{\Sigma,Z} \| P_\Sigma P_Z) \le \HH(P_Z).
\]
From \ref{prope:entropy bound}, we get $\HH(P_Z) \le \log|\Supp(\mathcal{Z})| \leq \log(2) S$, which completes the proof.
\end{proof}

\medskip 

We will consider the success probability of $(A_0,A_1)$ when some of the random variables are fixed and denote the fixed values in subscript. 
In particular, let
$p_z := \E_{P_{\Sigma,\SECV \mid Z=z}}[\SUC_{(A_1)_{z}^{\mathcal{O}(\Sigma,\SECV)}}(\SECV)]$ 
and 
$q_z := \E_{Q_{\Sigma,\SECV \mid Z=z}}[\SUC_{(A_1)_{z}^{\mathcal{O}(\Sigma,\SECV)}}(\SECV)]$ 
denote the success probabilities of $(A_0,A_1)$ conditioned on $Z=z$ with respect to $P$ and $Q$, respectively. 
Clearly, $p = \E_{P_Z(z)} [p_z]$.

The heart of the proof of Lemma~\ref{lem:non-fixed g-nonadaptive} is the following claim.
\begin{claim} \label{clm:non-fixed main-eq2}
For any $z$ in the support of $P_{Z}$,
\begin{align*} 
p_{z}
\leq 
\min \left(2 q_{z} 
+ 
\frac{4 \kappa_{z} T }{u},
q_{z} + \sqrt{\frac{\kappa_{z} T }{u}} \right).
\end{align*}
\end{claim}
The proof of Lemma~\ref{lem:non-fixed g-nonadaptive} below simply averages both sides of this claim over $P_{Z}(z)$.

\medskip

\begin{proof}[of Claim~\ref{clm:non-fixed main-eq2}]
We fix $z$, which fixes the queries of $A_1$.

Recall that $\mathcal{U} = \mathcal{U}(z)$ denotes the set of outer queries made by $A_1$, which includes all queries in this case. Define $|\SECS|$ sets $\{\mathcal{U}_\secv\}_{\secv \in \SECS}$,
where each $\mathcal{U}_\secv = \{ \TR(\secv,m) \mid m \in \mathcal{U}\}$ represents the set of indices of $\sigma$ queried by $A_1$ after translation, given that the secret is $\secv$ and the preprocessing string is $z$.
Fix any $i \in [N]$.

Since $\TR$ is a $u$-uniform translation function, we have
\begin{align} \label{eq:non-fixed readt2}
|\{\secv \mid i \in \mathcal{U}_\secv\}|
= | \cup_{m \in \mathcal{U}} \{ \secv \mid \TR(\secv,m) = i \} |
\leq \sum_{m \in \mathcal{U}} | \{ \secv \mid \TR(\secv,m) = i \} |
\le \frac{|\mathcal{U}| \cdot |\SECS|}{u} = \frac{T \cdot |\SECS|}{u}.
\end{align}

For every $\secv \in \SECS$, define the indicator function $f_\secv:[N]^{N} \to \{0,1\}$
as 
$f_\secv(\sigma) = 
\SUC_{(A_1)_{z}^{\mathcal{O}(\sigma,\secv)}}(\secv)$,
which outputs $1$ if
(the deterministic algorithm) $A_1$ succeeds on input $\sigma$ and secret $\secv$ when hard-wired with $z$. 

Recall that 
$p_{z,\secv} := \E_{P_{\Sigma \mid Z=z, \SECV = \secv}}[\SUC_{(A_1)_{z}^{\mathcal{O}(\Sigma,\secv)}}(\secv)]$, 
$q_{z,\secv} := 
\E_{Q_{\Sigma \mid Z=z, \SECV = \secv}}[\SUC_{(A_1)_{z}^{\mathcal{O}(\Sigma,\secv)}}(\secv) ]$
denote the success probabilities of $(A_0,A_1)$ conditioned on $Z=z$ and $\SECV=\secv$, with respect to $P$ and $Q$, respectively. Note that under $Q$, generally $z \neq A_0(\sigma)$,
but the success predicate $\SUC$ does not depend on $z$ (see Remark~\ref{remark:z}) and remains accurate.

Thus,
$$ 
\E_{   P_{\Sigma \mid Z = z} } [f_\secv(\Sigma)] = p_{z,\secv},\quad \E_{   Q_{\Sigma \mid Z = z} } [f_\secv(\Sigma)] = q_{z,\secv},
$$
and
$$ 
\frac{1}{|\SECS|} \sum_{\secv \in \SECS} 
p_{z,\secv} = p_z,\quad \frac{1}{|\SECS|} \sum_{\secv \in \SECS} 
q_{z,\secv} = q_z.
$$
Recall that for every $\secv \in \SECS$, $f_\secv(\sigma) = \SUC_{(A_1)_{z}^{\mathcal{O}(\sigma,\secv)}}(\secv)$
depends only on the secret and the part of $\sigma$ queried by $(A_1)_{z}^{\mathcal{O}(\sigma,\secv)}$. Hence, it follows from~(\ref{eq:non-fixed readt2}) that
$\{ f_\secv \}_{\secv \in \SECS}$ form a read-$(T \cdot |\SECS|/u)$ family on the $N$ variables of $\Sigma$.

Observe that $Q_\Sigma$ is the uniform distribution function over bijections on $[N]$, and $P_\Sigma$ is another distribution function over such bijections.

We apply Theorem~\ref{thm:Gavinsky-perm} and deduce  
\begin{align*}
(2 \cdot T \cdot |\SECS|/u) \cdot  \kappa_z  
= 
(2 \cdot T \cdot |\SECS|/u) \cdot \KL( P_{\Sigma \mid Z = z} \| Q_{\Sigma \mid Z = z} ) 
\geq 
|\SECS| \cdot 
\KL( p_z \| 
     q_z ).
\end{align*}
Therefore,
$\KL( p_z \| q_z ) \leq \frac{2 \cdot \kappa_z T}{u}$.

\medskip \noindent Applying Corollary~\ref{cor:KL-analysis} and~\ref{prope:pinsker}, respectively, we conclude 
$$p_{z} \leq 2 q_{z} + \frac{4 \kappa_{z} T }{u},
\textup{ and }
p_{z} \leq q_{z} + \sqrt{\frac{\kappa_{z} T }{u}},$$
as claimed.
\end{proof}

\medskip 

Finally, we prove Lemma~\ref{lem:non-fixed g-nonadaptive}.

\medskip

\begin{proof} [of Lemma~\ref{lem:non-fixed g-nonadaptive}]
We average both sides of the inequality in Claim~\ref{clm:non-fixed main-eq2} 
over $P_{Z}(z)$.

On the left-hand-side we obtain $\E_{P_{Z}(z)}
[p_{z}]=p,$ which is the success probability of $(A_0,A_1)$ with respect to $P$.

On the right-hand-side,
we consider $q_z$ and $\kappa_z$ separately.
First, by Claim~\ref{clm:non-fixed basic}, we obtain
$$
\E_{P_{Z}(z)}[q_z] = \E_{Q_{Z}(z)}[q_z] = q \leq \MAXS(T).$$

Second, recall that by Claim~\ref{clm:non-fixed klz}, we have $\E_{P_{Z}(z)}[\kappa_{z}] \leq \log(2) S$. Therefore, for the first term we obtain
\[
    \E_{P_{Z}(z)}\left[\frac{4 \kappa_{z} T }{u}\right] \le \frac{\log(2) 4 S T }{u}.
\]
For the second term, we use the concavity of the square root function to conclude
\[
    \E_{P_{Z}(z)}\left[\sqrt{\frac{\kappa_{z} T }{u}}\right] \le 
    \sqrt{\frac{\E_{P_{Z}(z)}[\kappa_{z}] T }{u}} \le
    \sqrt{\frac{\log(2) S T }{u}},
\]
completing the proof.
\end{proof}

\subsection{Proof of Theorem~\ref{thm:g-nonadaptive}}
\label{subsec:proof of strong version}
We now assume that $A_1$ makes $T_1 \geq 0$ inner queries and prove Theorem~\ref{thm:g-nonadaptive}.
The structure of the proof is similar to that of Lemma~\ref{lem:non-fixed g-nonadaptive}, albeit somewhat more involved. 

We prove Theorem~\ref{thm:g-nonadaptive} by a hybrid argument. We refer to the hybrid as a middle game:
\begin{definition}\label{def:MID}
    A middle game $\MID$ has the same setting as $\PC$, and it differs only in the allowed actions for the adversary. An adversary $\widehat{A}$ is a pair of algorithms $(\widehat{A}_0, \widehat{A}_1)$ that work as follows: $\widehat{A}_0$ receives only $N$ as input, and outputs \emph{constraints} which include a pair of sequences $\INP = (\INP_1,\dots,\INP_{T_1})$ and $\OUTP = (\OUTP_1,\dots,\OUTP_{T_1})$, each containing non-repeating elements from $[N]$. Then a permutation $\sigma:[N] \to [N]$ is sampled uniformly from the set of permutations satisfying $\sigma(\INP_j) = \OUTP_j$ for all $j$. $\widehat{A}_1$ is defined similarly to the algorithm $A_1$ from $\PC$, but with respect to that permutation $\sigma$, and \emph{it is allowed to make only outer queries and no inner queries}. The outer queries of $\widehat{A}_1$ may depend on $\INP$ and $\OUTP$.
    The rest of the game is defined similarly. The running time $T_2$ of $\widehat{A}_1$ is defined as the number of queries it makes, and the total running time of $\widehat{A}$ is defined as $T := T_1 + T_2$.
\end{definition}

\subsubsection{Proof overview.}
$\MID$ can be thought of as a non-preprocessing game, 
where the adversary is allowed to choose a limited number of $T_1$ values of the permutation (while the other values are sampled uniformly).
The proof is divided into two steps.
\begin{enumerate}
  \item The first step shows (in Theorem~\ref{thm:pc to middle} below)
        that the preprocessing model $\PC$ is not much stronger than $\MID$.
        First, allowing to fix $T_1$ values of the permutation
        gives the adversary in $\MID$ enough power to simulate the inner queries of the online algorithm in $\PC$,
        that depend on the preprocessing string.
        Once we have dealt with the inner queries, it remains to deal with the outer queries.
        We exploit the fact that in $\MID$ the non-fixed elements of the permutation are sampled uniformly
        from all the possible (remaining) values, hence we are left with a uniform permutation on a smaller space. 
        Thus, the proof proceeds in similar way to the proof of Lemma~\ref{lem:non-fixed g-nonadaptive}
        (where the online adversary only makes outer queries),
        considering the permutation on the smaller space. 
        As Lemma~\ref{lem:non-fixed g-nonadaptive}, Theorem~\ref{thm:pc to middle} is applicable to an arbitrary success predicate $\SUC$, as defined in Section~\ref{sec:PC}. 
  \item The second step shows (in Lemma~\ref{lem:middle to non preprocess} below) that fixing $T_1$ values of the permutation (independently of the secret) does not give the $\MID$ adversary much advantage over a non-preprocessing adversary.
        This lemma assumes that $\SUC$ compares the output of $(A_1)_z$ to some function of the secret.

    The proof proceeds by transforming any adversary for the middle game into a standard non-preprocessing, non-adaptive adversary for the original permutation challenge, while losing only a small additive term in success probability. The simulator runs the $\MID$ adversary, ignores its constraints, and simply queries the real oracle on the same outer queries. A coupling argument shows that the two experiments agree except for two bad events that may cause divergence. These events correspond to collisions with the constrained input and output sets, and can be bounded by $O(T^2/u)$.
\end{enumerate}
Combining the two above steps allows showing that in our setting, a preprocessing adversary is not much stronger than a non-preprocessing adversary,
proving Theorem~\ref{thm:g-nonadaptive}.

\subsubsection{First step.}
We first prove that the game $\MID$ is not much weaker than $\PC$.
\begin{theorem} \label{thm:pc to middle}
Let $\PC := \PC(N,\SECS,\mathcal{M},\TR,\POST,\SUC)$ be a permutation challenge game, such that 
$\TR$ is $u$-uniform.
Let $(A_0,A_1)$ be an $(S,T)$ non-adaptive algorithm with preprocessing for $\PC$.
Denote by $\widehat{\MAXS}(T)$ the optimal success probability (with respect to $\SUC$) of an adversary to $\MID$ with running time $T$.
Then, the success probability of $(A_0,A_1)$ is at most
\[
\min \left(2\cdot \widehat{\MAXS}(T) + \frac{4\log(2) S T }{u}, \widehat{\MAXS}(T) + 
\sqrt{\frac{\log(2) ST} {{u}}}\right).
\]
\end{theorem}

In order to prove the theorem, some preparation is needed.

\medskip 

First, let us recall a few notations we shall use in the proof. The set of inner queries is denoted by $\mathcal{S}$, and we denote $\overline{\mathcal{S}} = [N] \setminus \mathcal{S}$.  Recall that $\mathcal{S}$ depends on $z$.
The distribution $P$ was defined in Section~\ref{subsec:preprocess setup}, as $P_{\Sigma,Z,\SECV} = P_{\Sigma} P_{Z \mid \Sigma} P_\SECV$, where $P_\Sigma$ is the uniform distribution over permutations, $P_\SECV$ is the distribution over secrets $\SECV$, and $P_{Z \mid \Sigma = \sigma}(z) = 1$ if 
$z = A_0(\sigma)$.

We define a distribution $Q'_{\Sigma,Z,\SECV}$ which is identical to $P_{\Sigma,Z,\SECV}$,
with the exception that 
$$
Q'_{\Sigma \mid Z} 
= 
Q'_{\mathcal{S},\Sigma_{\mathcal{S}}, \Sigma_{\overline{\mathcal{S}}} \mid Z} 
=
Q'_{\mathcal{S},\Sigma_{\mathcal{S}} \mid Z} Q'_{\Sigma_{\overline{\mathcal{S}}} \mid Z, \mathcal{S},\Sigma_{\mathcal{S}} }
$$ 
is defined by setting
$Q'_{\mathcal{S},\Sigma_{\mathcal{S}} \mid Z} = P_{\mathcal{S},\Sigma_{\mathcal{S}} \mid Z}$,
while 
$Q'_{\Sigma_{\overline{\mathcal{S}}} \mid Z, \mathcal{S},\Sigma_{\mathcal{S}} }$
is the uniform distribution over bijections mapping the elements of $\overline{\mathcal{S}}$ to those of $[N] \setminus \Sigma_{\mathcal{S}}$.
On the other hand, $Q'_{\Sigma} = P_{\Sigma}$, $Q'_{Z} = P_{Z}$, $Q'_{\SECV} = P_{\SECV}$, and 
$Q'_{\Sigma,Z,\SECV} = Q'_{\Sigma,Z} \times Q'_\SECV$.

To simplify notation, 
denote the random variables that $A_1$ obtains from inner queries to $\sigma$
by 
$\Sigma^{in} := (\mathcal{S},\Sigma_{\mathcal{S}})$.

Given $\Sigma^{in}$, the distribution function $Q'_{\Sigma_{\overline{\mathcal{S}}} \mid Z, \Sigma^{in} }$ does not depend on the value of $Z$.
Therefore,
$$Q'_{\Sigma_{\overline{\mathcal{S}}} \mid Z, \Sigma^{in} } = Q'_{\Sigma_{\overline{\mathcal{S}}} \mid \Sigma^{in} }.$$

Note that if $T_1 = 0$ (i.e., if $\mathcal{S} = \emptyset$), then $Q'$ is identical to $Q$ defined above. However, in general we have
$
Q'_{\Sigma \mid Z} = 
Q'_{\Sigma^{in} \mid Z} 
Q'_{\Sigma_{\overline{\mathcal{S} }} \mid Z, \Sigma^{in} } =
P_{\Sigma^{in} \mid Z}
Q'_{\Sigma_{\overline{\mathcal{S} }} \mid \Sigma^{in}},
$
while
$Q_{\Sigma \mid Z} = 
Q_{\Sigma} =
P_{\Sigma}$.

Denote the success probability of $(A_0,A_1)$ with respect to $Q'$ by
$$
q' := \E_{Q'_{\Sigma,Z,\SECV}}
[ \SUC_{(A_1)_{Z}^{\mathcal{O}(\Sigma,\SECV)}}(\SECV) ].
$$

\begin{claim} We have
\label{clm:basic with advice}
$$
q' 
\leq 
\widehat{\MAXS}(T).$$
\end{claim}

\begin{proof}
Recall that $Q'_{\Sigma_{\overline{\mathcal{S}}} \mid Z, \Sigma^{in} } =
Q'_{\Sigma_{\overline{\mathcal{S}}} \mid  \Sigma^{in} }$
is the uniform bijection mapping the elements of $\overline{\mathcal{S}}$ to those of $[N] \setminus \Sigma_{\mathcal{S}}$.

Since $(A_1)_z$ queries $\mathcal{S}$, the preprocessing string $z$ does not give $A_1$ any additional information about 
$\Sigma$. Thus, $A$ under $Q'$ cannot do better than in the following game: First $A_0$ receives a uniform permutation $\sigma'$, and chooses a set $\mathcal{S}$ of elements to fix to pre-chosen values. Then, these elements are fixed, and the other elements of $\sigma'$ are re-shuffled to obtain $\sigma$. The rest of the game occurs with respect to $\sigma$. One can see that this game is equivalent to $\MID$.
 
Formally, given a preprocessing algorithm $(A_0,A_1)$ we define a (randomized) adversary $(\widehat{A}_0, \widehat{A}_1)$ for $\MID$
that makes $T$ queries and has success probability identical to $(A_0,A_1)$ (i.e., $q'$). By definition, 
the success probability of $(\widehat{A}_0, \widehat{A}_1)$ is bounded by $\widehat{\MAXS}(T)$, hence this will complete the proof.
It remains to define $(\widehat{A}_0, \widehat{A}_1)$. 

$\widehat{A}_0$ receives $N$ as input, samples a permutation $\sigma'$, and calls $A_0$ to obtain $z$. Then, it calls $A_1$ with $z$ as input and obtains the set $\mathcal{S}(z)$ of inner queries. Finally, $\widehat{A}_0$ outputs $\INP = \mathcal{S}(z)$ and $\OUTP = \sigma'_{\mathcal{S}(z)}$. $\widehat{A}_1$ receives oracle access to a permutation $\sigma$ with $\sigma_{\mathcal{S}} = \sigma'_{\mathcal{S}}$ and independent values elsewhere. Notably, $\mathcal{S},\sigma_{\mathcal{S}}$ and $z$ are known to $\widehat{A}_1$, since $\widehat{A}_0$ and $\widehat{A}_1$ share the randomness tape and it is the only randomness source of $\widehat{A}_0$. Therefore, the distribution of $\sigma$ conditioned on $z$ is exactly $Q'_{\Sigma \mid Z}$. Therefore, $\widehat{A}_1$ can perfectly simulate $A_1$ by making the same (non-adaptive) outer queries as $A_1$ 
and returning its output. Thus, the success probability of $\widehat{A}$ is $q'$, as required.
\end{proof}

Fixing $z,\sigma^{in}$,
we can think of $\sigma_{\overline{\mathcal{S}}}$ drawn from either 
$P_{ \Sigma_{\overline{\mathcal{S}}} \mid  \Sigma^{in} = \sigma^{in}, Z = z}$
or 
$Q'_{ \Sigma_{\overline{\mathcal{S}}} \mid \Sigma^{in} = \sigma^{in}}$
as an extended bijection from $[N]$ to itself that maps the elements of $\mathcal{S}$ to those in $\sigma_{\mathcal{S}}$,
and maps the elements of $\overline{\mathcal{S}}$ to those of $[N] \setminus \sigma_{\overline{\mathcal{S}}}$.

We can therefore think of $\sigma^{in}$ as being hardwired into $(A_1)_{z}$ such that it only makes queries to the random variables 
$\Sigma_{\mathcal{S}}$ (after translation). Of course, some of the (translated) queries of $(A_1)_{z}$ may still fall in $\mathcal{S}$, but 
$\sigma^{in} = (\mathcal{S},\sigma_{\mathcal{S}})$ are formally no longer random variables.

Denote
$$p_{z, \sigma^{in} } := 
\E_{P_{ \Sigma_{\overline{\mathcal{S}} },\SECV \mid \Sigma^{in} = \sigma^{in}, Z=z}}
[ \SUC_{(A_1)_{z}^{\mathcal{O}( \sigma^{in}, \Sigma_{\overline{\mathcal{S}}},\SECV)}}(\SECV)
] $$ 
and
$$q'_{z, \sigma^{in} } := \E_{Q'_{ \Sigma_{\overline{\mathcal{S}} },\SECV \mid  \Sigma^{in} = \sigma^{in}, Z=z}}
[ \SUC_{(A_1)_{z}^{\mathcal{O}( \sigma^{in}, \Sigma_{\overline{\mathcal{S}}},\SECV)}}(\SECV)
].
$$ 

Also, denote
$$\kappa'_{z, \sigma^{in} } = 
\KL( P_{ \Sigma_{\overline{\mathcal{S}}} \mid \Sigma^{in} = \sigma^{in} , Z = z} \| Q'_{ \Sigma_{ \overline{\mathcal{S}} } \mid \Sigma^{in} = \sigma^{in} } ). 
$$

The heart of the proof of Theorem~\ref{thm:pc to middle} is the following additional claim. 
\begin{claim} \label{clm:advice main-eq2}
For any $z,\sigma^{in}$ in the support of $P_{Z,\Sigma^{in}}$,
\begin{align*} 
p_{z, \sigma^{in} }
\leq 
\min \left(
2 q'_{z, \sigma^{in} }
+ 
\frac{4 \kappa'_{z, \sigma^{in} } T_2 }{u},
q'_{z,\sigma^{in} } + \sqrt{\frac{\kappa'_{z,  \sigma^{in} } T_2} {{u}}} 
\right).
\end{align*}
\end{claim}

\begin{proof}
We fix $z,\sigma^{in}$ (hence fixing the queries of $A_1$). 

As in the proof of Claim~\ref{clm:non-fixed main-eq2}, 
define $|\SECS|$ sets $\{\mathcal{U}_\secv\}_{\secv \in \SECS }$,
where $\mathcal{U}_\secv = \{ \TR(\secv,m) \mid m \in \mathcal{U}\}$.
Recall that $\Sigma^{in}$ and $\SECV$ are independent random variables (under both $P$ and $Q'$),
so for any value $\Sigma^{in} = \sigma^{in}$, $\SECV$ is still uniform.

For every $\secv \in \SECS$, define the indicator function $f_c:[N]^{N - T_1} \mapsto \{0,1\}$ by setting
$f_\secv(\sigma_{\overline{\mathcal{S}}}) = 
\SUC_{(A_1)_{z}^{\mathcal{O}(\sigma^{in} , \sigma_{\overline{\mathcal{S}}},\secv)}}(\secv),
$
which outputs $1$ if
(the deterministic algorithm) $A_1$ succeeds with the secret $\secv$ on input $\sigma_{\overline{\mathcal{S}}}$ when hard-wired with $z,\sigma^{in}$. 

Recall that
$$p_{z,\sigma^{in},\secv} := 
\E_{P_{ \Sigma_{\overline{\mathcal{S}}} \mid \Sigma^{in} = \sigma^{in}, Z=z, \SECV = \secv}}
[ 
\SUC_{(A_1)_{z}^{\mathcal{O}( \sigma^{in} , \Sigma_{\overline{\mathcal{S}}},\secv)}}(\secv)
]$$ 
and
$$q'_{z,\sigma^{in},\secv} := \E_{Q'_{ \Sigma_{\overline{\mathcal{S}}} \mid \Sigma^{in} = \sigma^{in}, Z=z, \SECV = \secv}} [
\SUC_{(A_1)_{z}^{\mathcal{O}(\sigma^{in}, \Sigma_{\overline{\mathcal{S}}},\secv)}}(\secv)
]$$ 
denote the success probability of $(A_0,A_1)$ conditioned on $Z=z,\Sigma^{in} = \sigma^{in}, \SECV = \secv$ with respect to $P$ and $Q'$, respectively.

Thus,
$$ 
\E_{   P_{\Sigma_{\overline{\mathcal{S}}} \mid \Sigma^{in} = \sigma^{in} ,Z = z} } [f_d(\Sigma_{\overline{\mathcal{S}}})] 
= 
p_{z,\sigma^{in},\secv}
\qquad \mbox{and} \qquad
\E_{ Q_{\Sigma_{\overline{\mathcal{S}}} \mid \Sigma^{in} = \sigma^{in} ,Z = z}} 
[f_d(\Sigma_{\overline{\mathcal{S}}})] 
= 
q'_{z,\sigma^{in},\secv}.
$$
Also,
$$ 
\frac{1}{|\SECS|} \sum_{\secv \in \SECS} 
p_{z,\sigma^{in},\secv}
= 
p_{z,\sigma^{in}}
\qquad \mbox{and} \qquad 
\frac{1}{|\SECS|} \sum_{\secv \in \SECS} 
q'_{z,\sigma^{in},\secv}
= 
q'_{z,\sigma^{in} }.
$$

Recall that when $\sigma^{in}$ is fixed, $A_1$ only makes queries in
$\overline{\mathcal{S}}$ (after translation), and by similar calculation to~(\ref{eq:non-fixed readt2}),
$\{ f_\secv \}_{\secv \in \SECS}$ form a read-$(T_2 \cdot |\SECS|/u)$ family on the $N - T_1$ variables of $\Sigma_{\overline{\mathcal{S}}}$. 

Note that $\sigma_{\overline{\mathcal{S}}}$ drawn from either 
$P_{ \Sigma_{\overline{\mathcal{S}}} \mid \Sigma_{\mathcal{S}} = \sigma_{\mathcal{S}}, Z = z}$
or 
$
Q'_{ \Sigma_{\overline{\mathcal{S}}} \mid \Sigma^{in} = \sigma^{in}, Z = z} =
Q'_{ \Sigma_{\overline{\mathcal{S}}} \mid \Sigma^{in} = \sigma^{in} }$
is a bijection from $\overline{\mathcal{S}}$ to $[N] \setminus \sigma_{\mathcal{S}}$,
and 
$Q'_{ \Sigma_{\overline{\mathcal{S}}} \mid \Sigma^{in} = \sigma^{in} }$
is uniformly chosen.
We apply Theorem~\ref{thm:Gavinsky-perm} and deduce  
\begin{align*}
\begin{split}
(2 \cdot T_2 \cdot |\SECS|/u) \kappa'_{z, \sigma^{in} }  
&= 
(2 \cdot T_2 \cdot |\SECS|/u)  \KL( P_{ \Sigma_{\overline{\mathcal{S}}} \mid \Sigma^{in} = \sigma^{in}, Z = z} \| Q'_{ \Sigma_{ \overline{\mathcal{S}} } \mid \Sigma^{in} = \sigma^{in} } )  \\
&\geq 
|\SECS|
\KL( p_{z, \sigma^{in} } \| 
     q'_{z,\sigma^{in} } ).
\end{split}
\end{align*}
Therefore,
$$\KL( p_{z,\sigma^{in}} \| 
     q'_{z,\sigma^{in} } ) \leq \frac{2 \cdot \kappa'_{z, \sigma^{in} } T_2}{u}.$$
Applying Corollary~\ref{cor:KL-analysis} and~\ref{prope:pinsker}, respectively, we conclude 
$$p_{z,\sigma^{in} } \leq 2 q'_{z,\sigma^{in} } + \frac{4 \kappa'_{z,  \sigma^{in} } T_2 }{u} \qquad
\textup{ and } \qquad 
p_{z,\sigma^{in} } \leq q'_{z,\sigma^{in} } + \sqrt{\frac{\kappa'_{z,  \sigma^{in} } T_2} {{u}}},
$$
as claimed.
\end{proof}

\medskip Now, we are ready to prove Theorem~\ref{thm:pc to middle}.

\medskip
\begin{proof} [of Theorem~\ref{thm:pc to middle}]
We average both sides of the inequality of Claim~\ref{clm:advice main-eq2} 
over $P_{Z,\Sigma^{in} }(z,\sigma^{in}) = Q'_{Z,\Sigma^{in} }(z,\sigma^{in})$.

On the left-hand-side, we obtain
$
\E_{P_{Z,\Sigma^{in}}(z, \sigma^{in} )}
[p_{z,\sigma^{in} }]
= p$.

On the right-hand-side,
we consider $q'_{z, \sigma^{in}}$ and $\kappa'_{z,\sigma^{in} }$ separately.

First, since
$P_{Z,\Sigma^{in} }(z,\sigma^{in}) = Q'_{Z, \Sigma^{in} }(z,\sigma^{in})$, we have
$$
\E_{Q'_{Z,\Sigma^{in}}(z,\sigma^{in})}[
q'_{z, \sigma^{in}} ] =
q' \leq 
\widehat{\MAXS}(T),$$
where the inequality is by Claim~\ref{clm:basic with advice}.

Second, we have
\begin{align*} 
\E_{Q_{Z,\Sigma^{in}}(z, \sigma^{in} )} [ \kappa'_{z,\sigma^{in} } ] 
&=  
\E_{Q_{Z,\Sigma^{in}}(z, \sigma^{in} )} [ \KL( P_{ \Sigma_{\overline{\mathcal{S}}} \mid \Sigma^{in} = \sigma^{in} , Z = z} \| Q'_{ \Sigma_{ \overline{\mathcal{S}} } \mid \Sigma^{in} = \sigma^{in} } ) ]. 
\end{align*}
Recall that 
for any $\sigma^{in}$
$,
Q'_{ \Sigma_{ \overline{\mathcal{S}} } \mid \Sigma^{in} = \sigma^{in} }$
is uniform 
over its support, hence, by~\ref{prope:kl and entropy}
\begin{align*} 
& \E_{Q_{Z,\Sigma^{in}}(z, \sigma^{in} )} [ \KL( P_{ \Sigma_{\overline{\mathcal{S}}} \mid \Sigma^{in} = \sigma^{in} , Z = z } \| Q'_{ \Sigma_{ \overline{\mathcal{S}} } \mid \Sigma^{in} = \sigma^{in} } ) ] \\
&= 
\E_{Q_{Z,\Sigma^{in}}(z, \sigma^{in} )}[\HH(Q'_{\Sigma_{ \overline{\mathcal{S}} } \mid \Sigma^{in} = \sigma^{in}} ) - 
\HH(P_{ \Sigma_{\overline{\mathcal{S}}} \mid \Sigma^{in} = \sigma^{in} , Z = z } ) ] \\
&= 
\HH(Q'_{\Sigma_{ \overline{\mathcal{S}} } \mid \Sigma^{in}} ) - 
\HH(P_{ \Sigma_{\overline{\mathcal{S}}} \mid \Sigma^{in} , Z} )
= 
\HH(Q'_{\Sigma_{ \overline{\mathcal{S}} } \mid \mathcal{S},\Sigma_{\mathcal{S} },Z } ) - 
\HH(P_{ \Sigma_{\overline{\mathcal{S}}} \mid \mathcal{S},\Sigma_{\mathcal{S} },Z} ) \\
&=
\HH(Q'_{\Sigma_{ \overline{\mathcal{S}} } \mid \Sigma_{\mathcal{S}},Z } ) - 
\HH(P_{ \Sigma_{\overline{\mathcal{S}}} \mid \Sigma_{\mathcal{S}},Z} ).
\end{align*} 
By~\ref{prope:chain rule entropy} and~\ref{prope:conditioned entropy},
\begin{align*} 
\HH(Q'_{\Sigma_{ \overline{\mathcal{S}} } \mid \Sigma_{\mathcal{S}},Z } ) 
&= 
\HH(Q'_{ \Sigma \mid Z} ) - 
\HH(Q'_{ \Sigma_{\mathcal{S}} \mid Z}) = 
\HH(Q'_{ \Sigma \mid Z} ) - 
\HH(P_{ \Sigma_{\mathcal{S}} \mid Z}) \leq
\HH(Q'_{ \Sigma} ) - 
\HH(P_{ \Sigma_{\mathcal{S}} \mid Z}) \\
&\leq 
\HH(P_{ \Sigma} ) - 
\HH(P_{ \Sigma_{\mathcal{S}} \mid Z}),
\end{align*} 

and similarly,
$$
\HH(P_{ \Sigma_{\overline{\mathcal{S}}} \mid \Sigma_{\mathcal{S}},Z} ) = 
\HH(P_{ \Sigma \mid Z} ) - 
\HH(P_{ \Sigma_{\mathcal{S}} \mid Z}). 
$$
Therefore, by~\ref{prope:information bound},~\ref{prope:entropy bound},
and since $(A_0,A_1)$ is an $(S,T)$ algorithm,
\begin{align*} 
\E_{Q_{Z,\Sigma^{in}}(z, \sigma^{in} )} [ \kappa'_{z,\sigma^{in} } ] \leq
\HH(P_{\Sigma}) - \HH(P_{\Sigma \mid Z} ) = \I(\Sigma ; Z) \leq
\HH(Z) \leq 
\log|\Supp(\mathcal{Z})| \leq \log(2) S.
\end{align*}

Overall, we obtain 
$$ 
p \leq 
2 \cdot \widehat{\MAXS}(T) + \frac{\log(2) 4 S T }{u},
\textup{ and }
p \leq 
\widehat{\MAXS}(T) + 
\sqrt{\frac{\log(2) ST} {{u}}},
$$
where the second inequality uses the concavity of the square root function.
This concludes the proof.
\end{proof}

\subsubsection{Second step - completing the proof of Theorem~\ref{thm:g-nonadaptive}.}
Theorem~\ref{thm:g-nonadaptive} follows immediately from Theorem~\ref{thm:pc to middle} and the following lemma
(which is the only lemma that restricts $\SUC$):
\begin{lemma}\label{lem:middle to non preprocess}
    Let $\PC := \PC(N,\SECS,\mathcal{M},\TR,\POST,\SUC)$ be a permutation challenge game, such that $\TR$ is $u$-uniform,
    and $\SUC$ compares the output of $A_1$ to some function of the secret.
    Denote by $\MAXS(T)$ the optimal success probability (with respect to $\SUC$) of a non-preprocessing, non-adaptive algorithm that makes at most $T$ queries.

    Let $\MID$ be the corresponding game defined above (see Definition~\ref{def:MID}).
    Denote by $\widehat{\MAXS}(T)$ the optimal success probability (with respect to $\SUC$) of an adversary to $\MID$ with running time $T$. Then $$\widehat{\MAXS}(T) \le \MAXS(T) + \frac{T^2}{2u}.$$
\end{lemma}
\begin{proof}
    Let $\widehat{A} = (\widehat{A}_0, \widehat{A}_1)$ be an adversary to $\MID$ with running time $T = T_1 + T_2$, and denote its success probability by $\widehat{q}$. We describe a non-preprocessing algorithm $A$ for $\PC$ with success probability $q \ge \widehat{q} - \frac{T^2}{2u}$ that makes $T$ queries.
     
     $A$ simulates $\widehat{A}_0$ to obtain the constraints $\INP,\OUTP$, which also determine the outer queries of $\widehat{A}_1$ (recall that a $\MID$ game has no inner queries). It ignores the constraints of $\widehat{A}$, requests its outer queries from $\calO$, sends the responses to $\widehat{A}_1$ and outputs the same value.

    Let $g: \SECS \to \SECS'$
    be the function of the secret computed by $\SUC$ ($\SECS'$ is its output space).
    Denote by $\textup{P}: \SECS' \times \SECS' \to \{0,1\}$ the success predicate of $\SUC$.

    Below, we define the algorithm $\SUC_{A^{ \mathcal{\calO}(\sigma,\secv)} }(\secv)$ 
    in the above case, where $\widehat{A}$ is simulated by $A$ (we omit the explicit interaction with $A$ for simplicity).
    \begin{enumerate}
        \item Receive $\INP,\OUTP$ from $\widehat{A}_0$.
        \item Pass $\INP,\OUTP$ to $\widehat{A}_1$ and receive the outer queries $\mathcal{U} = (m_1,\dots,m_{T_2})$.
        \item Sample $\secv \sim \SECS$.
        \item For all $1 \le i \le T_2$:
        \begin{enumerate}
            \item Denote $u_i \leftarrow \TR(\secv,m_i)$.
            \item If $u_i = u_j$ for some $j < i$, set $v_i \leftarrow v_j$. 
            Otherwise, sample $v_i \sim [N]\setminus \{v_1,\dots,v_{i-1}\}$ uniformly.
        \end{enumerate}
        \item Pass $(\POST(\secv,v_1),\dots,\POST(\secv,v_{T_2}))$ to $\widehat{A}_1$ and receive the answer 
        $\tilde{\secv} \in \SECV'$.
        \item Return $\textup{P}(g(\secv),\tilde{\secv})$.
    \end{enumerate}
      Below, we define the algorithm $\SUC_{\widehat{A}_1^{ \mathcal{\calO}(\sigma,\secv)} }(\secv)$,
      in case that $\widehat{A}$ runs in the $\MID$ game.
     It uses two additional flags, $W_1,W_2$, that do not affect the output, and are only defined for the sake of the analysis. The differences from the previous algorithm are underlined. 
    \begin{enumerate}
    \item Receive $\INP,\OUTP$ from $\widehat{A}_0$.
        \item Pass $\INP,\OUTP$ to $\widehat{A}_1$ and receive the outer queries $\mathcal{U} = (m_1,\dots,m_{T_2})$.
        \item Sample $\secv \sim \SECS$.
        \item \underline{Set $W_1 \leftarrow 0, W_2 \leftarrow 0$.} 
        \item For all $1 \le i \le T_2$:
        \begin{enumerate}
            \item Denote $u_i \leftarrow \TR(\secv,m_i)$.
            \item \underline{If $u_i = \INP_k \in \INP$, set $v_i \leftarrow \OUTP_k$ and $W_1 \leftarrow 1$.
            Continue to $m_{i+1}$.}
            \item If $u_i = u_j$ for some $j < i$, set $v_i \leftarrow v_j$. 
            Otherwise, sample $v_i \sim [N]\setminus \{v_1,\dots,v_{i-1}\}$ uniformly.
            \item \underline{If $v_i \in \OUTP$, 
            re-sample $v_i \sim [N]\setminus (\OUTP \cup \{v_1,\dots,v_{i-1}\})$ uniformly, and set $W_2 \leftarrow 1$.}       
        \end{enumerate}
        \item Pass $(\POST(\secv,v_1),\dots,\POST(\secv,v_{T_2}))$ to $\widehat{A}_1$ and receive the answer 
        $\tilde{\secv} \in \SECV'$.
        \item Return $\textup{P}(g(\secv),\tilde{\secv})$.
    \end{enumerate}
    Observe that the algorithms are equivalent, as long as $W_1 = W_2 = 0$ at the end of the execution of $\SUC_{\widehat{A}_1^{ \mathcal{\calO}(\sigma,\secv)} }(\secv))$.
    Thus, 
    \begin{align*}
    q 
    & = \Pr[\SUC_{A^{ \mathcal{\calO}(\Sigma,\SECV)} }(\SECV)] \\
    & \geq \Pr[ \SUC_{\widehat{A}_1^{ \mathcal{\calO}(\Sigma,\SECV)} }(\SECV)) \cap (W_1 = W_2 = 0)]  \\
    & \geq \Pr[ \SUC_{\widehat{A}_1^{ \mathcal{\calO}(\Sigma,\SECV)} }(\SECV)) ] - \Pr[W_1 = 1] - \Pr[W_2 = 1] \\
    & = \widehat{q} - \Pr[W_1 = 1] - \Pr[W_2 = 1].
    \end{align*}
    By a union bound over the $T_2$ outer queries, we have $\Pr[W_1 = 1] \leq \frac{T_1 T_2}{u} \le \frac{T^2}{4u}$, where the probability is taken over the choice of 
    $d$. By another union bound over the $T_2$ outer queries, $\Pr[W_2 = 1] \leq \frac{T^2}{4u}$, 
    where the probability is taken over the choice of $c_1,c_2,\ldots$.
    Overall, $q \ge \widehat{q} - \frac{T^2}{2u}$, as claimed.
\end{proof}

\section{Applications to Problems in the Generic Group Model}
\label{sec:DLOG}

In this section we describe our applications to the DLOG, DDH, and sqDDH problems. All these applications (as well as the application to the Even-Mansour cryptosystem presented in Section~\ref{sec:EM}) are obtained via Theorem~\ref{thm:g-nonadaptive}.
This theorem can indeed be used, since in all these applications, $\SUC$ compares the output of $(A_1)_z$ to some function of the secret
(as defined in Section~\ref{sec:inst GGM}).
We stress again that none of the uses of this theorem directly analyzes the preprocessing setting. 
For simplicity, we count the input group elements of the generic group algorithm $A_1$ as part of its $T$ queries. Otherwise, we have to add to $T$ a small additive factor based on its input size.

\subsection{The discrete-log problem}

As was written in the introduction, the bound we prove for the DLOG problem is the following.

\medskip \noindent \textbf{Theorem~\ref{thm:dl-nonadaptive-intro}.}
Let $A=(A_0,A_1)$ be a non-adaptive $(S,T)$-algorithm for the DLOG problem in the generic group model, over a group $G$ with a prime number $N$ of elements. Denote by $\MAXS_{\textup{DLOG}}(T)$ 
the optimal success probability of a non-preprocessing, non-adaptive algorithm that makes at most $T$ queries.
Then, the success probability of $A$ is at most
$$
2\cdot \MAXS_{\textup{DLOG}}(T) + \frac{4\log(2) ST}{N} + \frac{T^2}{N} \leq
\frac{3T^2 }{N}  + \frac{4\log(2) ST}{N}.
$$

\begin{proof}[of Theorem~\ref{thm:dl-nonadaptive-intro}]
We would like to use Theorem~\ref{thm:g-nonadaptive} via the corresponding permutation challenge game
$\PC := \PC_{DL}(N,\SECS,\mathcal{M},\TR,\POST,\SUC)$
defined in Section~\ref{sec:inst GGM}. For this purpose it remains to show that 
$\TR$ is $N$-uniform, 
where $\TR(d,(a,b)) = a \cdot d + b \mod N$.
Since for every query the translation function defines a polynomial of degree $1$ in $\secv$,
Lemma~\ref{lem:ggm uniform} implies that $\TR$ is $N$-uniform.

Finally, as was mentioned in the introduction, the bound on $\textup{online}(A_1)$ on the right-hand-side is by Shoup's theorem~\cite{Shoup97} which bounds $\MAXS_{\textup{DLOG}}(T) \leq \frac{T^2}{N}$ for DLOG.
\end{proof}

\subsection{The DDH and sqDDH problems}

As was written in the introduction, the main theorem we prove for DDH and sqDDH is the following.

\medskip \noindent \textbf{Theorem~\ref{thm:DDH-intro1}.} Let $A=(A_0,A_1)$ be a non-adaptive $(S,T)$ algorithm for the DDH (resp.~sqDDH) problem in the GGM, over a group $G$ with a prime number $N$ of elements. Denote by $\MAXS_{\textup{DDH}}(T)$ the optimal success probability of a non-preprocessing, non-adaptive algorithm that makes at most $T$ queries. Then, the success probability of $A$ is at most
$$
\MAXS_{\textup{DDH}}(T) + \sqrt{\frac{2\log(2) ST}{N}} + \frac{T^2}{N} \leq
\frac{1}{2}+ \frac{2T^2}{N}  + \sqrt{\frac{2\log(2)ST}{N}}.
$$

\medskip

\begin{proof}[of Theorem~\ref{thm:DDH-intro1}]
We treat DDH and sqDDH separately.

\paragraph{DDH.} We would like to use Theorem~\ref{thm:g-nonadaptive} via the corresponding permutation challenge game
$\PC := \PC_{DDH}(N,\SECS,\mathcal{M},\TR,\POST,\SUC)$
defined in Section~\ref{sec:inst GGM}. For this purpose, it remains to show that 
$\TR$ is $(N/2)$-uniform.

Recall that 
\begin{align*}
\TR((d_1,d_2,d_3,k),(a_1,a_2,a_3,b) ) =
\begin{cases}
a_1 \cdot d_{1} + a_2 \cdot d_{2} + a_3 \cdot d_{3} + b \mod N & \text{if } k=0, \\
a_1 \cdot d_{1} + a_2 \cdot d_{2} + a_3 \cdot (d_{1} d_{2}) + b \mod N & \text{if } k = 1.
\end{cases}
\end{align*}
Since for every query, the translation function defines a polynomial of degree $1$ or $2$ in $c_1,c_2,c_3$,
Lemma~\ref{lem:ggm uniform} implies that $\TR$ is $(N/2)$-uniform, as asserted.

\paragraph{sqDDH.}
We apply Theorem~\ref{thm:g-nonadaptive} via the corresponding permutation challenge game
$\PC := \PC_{sqDDH}(N,\SECS,\mathcal{M},\TR,\POST,\SUC)$
defined in Section~\ref{sec:inst GGM}. For this purpose, it remains to show that 
$\TR$ is $(N/2)$-uniform.

Recall that 
\begin{align*}
\TR((d_1,d_2,k),(a_1,a_2,b) ) =
\begin{cases}
a_1 \cdot d_{1} + a_2 \cdot d_{2}  + b \mod N & \text{if } k=0, \\
a_1 \cdot d_{1} + a_2 \cdot (d_{1})^2 + b \mod N & \text{if } k = 1.
\end{cases}
\end{align*}
Since for every query, the translation function defines a polynomial of degree $1$ or $2$ in $d_1,d_2$,
Lemma~\ref{lem:ggm uniform} implies that $\TR$ is $(N/2)$-uniform, as asserted.

Finally, the inequality in the assertion holds by Shoup's theorem~\cite{Shoup97} which bounds $\MAXS_{\textup{DDH}}(T) \leq \frac{1}{2} +\frac{T^2}{N}$ for DDH, and a similar argument which yields the same bound for sqDDH (see~\cite{CorettiDG18}).
\end{proof}

\section{Application to the Even-Mansour Cryptosystem}
\label{sec:EM}

In this section we present our application to the EM cryptosystem. We first briefly review previous results on EM, and then we present the proof of Theorem~\ref{thm:em-nonadaptive-intro}. 

\subsection{Previous results on the security of EM}

Recall that the Even-Mansour cryptosystem~\cite{EvenM97} is defined as $EM(m)=k_2 \oplus \sigma(k_1 \oplus m)$, where $k_1,k_2 \in \{0,1\}^n$ are $n$-bit keys and $\sigma:\{0,1\}^n \to \{0,1\}^n$ is a publicly known permutation. Even and Mansour showed that if $\sigma$ is chosen uniformly at random, then any attack which makes $T_2$ encryption/decryption queries to $EM$ and $T_1$ queries to $\sigma$ or $\sigma^{-1}$, has a success probability of $O(T_1 T_2 / 2^n)$. 

Right after the introduction of the EM cryptosystem, Daemen~\cite{Daemen91} presented a matching attack: Pick an arbitrary $\alpha \in \{0,1\}^n$, query $\sigma$ to obtain $T_1/2$ pairs of values $(\sigma(x_i),\sigma(x_i \oplus \alpha))$, and store the pairs $(\sigma(x_i) \oplus \sigma(x_i \oplus \alpha), x_i)$ in a sorted table. Then, query $EM$ to obtain $T_2/2=N/T_1$ pairs of values $(EM(m_i),EM(m_i \oplus \alpha))$, and check, for each $j$, whether $EM(m_j) \oplus EM(m_j \oplus \alpha)$ appears in the table. If a collision of the form $EM(m_j) \oplus EM(m_j \oplus \alpha)=\sigma(x_i) \oplus \sigma(x_i \oplus \alpha)$ is found, then it is likely that $k_1 = m_j \oplus x_i$ or $k_1=m_j \oplus x_i \oplus \alpha$. At this stage, the second key $k_2$ can be retrieved easily. The attack algorithm is based on the fact that XOR with a secret key preserves XOR differences, and on the birthday paradox.
The memory complexity of Daemen's attack is $\tilde{O}(T_1)$. 

In 2000, Biryukov and Wagner~\cite{BW00} presented an alternative attack with complexities of $T_1=T_2=O(2^{n/2})$, and in 2012, Dunkelman, Keller and Shamir~\cite{DKS15} showed that its memory complexity can be reduced to $\tilde{O}(1)$, using an adaptive collision-search procedure. The authors of~\cite{DKS15} also generalized the attack of~\cite{BW00} to the entire tradeoff curve $T_1 T_2 = O(2^n)$ and asked whether its memory complexity can be reduced to $\tilde{O}(1)$ for the entire curve. This was partially addressed by Fouque, Joux and Mavromati~\cite{FouqueJM14}, who showed that the memory complexity can be reduced to $o(T_1)$ using an adaptive procedure, though a reduction to $\tilde{O}(1)$ has not been found yet.

While Daemen's attack is clearly non-adaptive, the Biryukov-Wagner attack and its enhancements heavily use adaptivity via Floyd's algorithm. Like in the case of DLOG, for about 35 years it hasn't been known whether adaptivity is indeed essential for reducing the memory complexity below the $\tilde{O}(T_1)$ complexity of Daemen's algorithm. 

Daemen's algorithm can be naturally viewed as an algorithm with preprocessing, as the public knowledge of $\sigma$ allows performing all queries to it offline. The Biryukov-Wagner algorithm and its variants do not use preprocessing. Fouque, Joux and Mavromati~\cite{FouqueJM14} showed that if preprocessing with a space of $S=2^{n/3}$ is allowed, then an attack with online complexity of $T=\tilde{O}(2^{n/3})$ can be mounted. In the other direction, Coretti, Dodis and Guo~\cite{CorettiDG18} showed that any distinguishing attack on EM with preprocessing has a success probability of $\frac{1}{2}+ \tilde{O}(\sqrt{S(T_1+T_2)T_2/2^n}+T_1 T_2/2^n)$, matching the attack of~\cite{FouqueJM14} in the case of constant success probability. 

\subsection{Our bound for EM}

In this section we prove Theorem~\ref{thm:em-nonadaptive-intro} stated in the introduction. We assume that adversary can query only the public and encryption oracles and not the decryption oracle. A direct corollary due to the symmetry of EM is that the same results hold if the adversary can query only the public and decryption oracles.

\medskip \noindent \textbf{Theorem~\ref{thm:em-nonadaptive-intro}.}
Let $A=(A_0,A_1)$ be a~key-recovery, non-adaptive $(S,T)$-adversary for the Even-Mansour cryptosystem, which can query only the public and encryption oracles and not the decryption oracle. Denote by $\MAXS_{\textup{EM}}(T)$ the optimal success probability of a non-preprocessing, non-adaptive algorithm that makes at most $T$ queries. Then, the success probability of $A$ is at most
$$
2\cdot \MAXS_{\textup{EM}}(T) + \frac{4\log(2) S (T+1)}{N} + \frac{T^2}{N} \leq
    \frac{3T^2 }{N}  + \frac{4\log(2) S(T+1)}{N}.
$$
Moreover, the theorem also holds for the single-key variant where $k_1=k_2$.

\medskip

\begin{proof}[of Theorem~\ref{thm:em-nonadaptive-intro}]
We would like to use Theorem~\ref{thm:g-nonadaptive} via the corresponding permutation challenge game
$\PC := \PC_{EM-KR}(N,\SECS,\mathcal{M},\TR,\POST,\SUC)$
defined in Section~\ref{sec:inst EM}. 
First, observe that $\SUC$ compares the output of $(A_1)_z$ to the secret key, which is a function of the secret.
It remains to show that
$\TR$ is $N$-uniform, 
where $\TR((k_1,k_2),m) = m \oplus k_1$.

Fix $m \in [N]$ and let $j \in [N]$.
Then, 
$$
\TR((k_1,k_2)),m) = j \Leftrightarrow m \oplus k_1 = j
\Leftrightarrow k_1 = j \oplus m.
$$
Therefore,
$$\E_{P_{K_1,K_2}} [\mathbbm{1}(\TR((K_1,K_2),m) = j)] = 
\E_{P_{K_1,K_2}} [\mathbbm{1}(K_1 = j \oplus m)] = 
\tfrac{1}{N},$$
as required.
Note that this also holds for the single-key scheme with $k_1 = k_2$. 
\end{proof}

\appendix

\section{Limitations of Previous Techniques in Dealing with Non-Adaptive Algorithms}
\label{app:limitations}

We consider the three main generic techniques for proving cryptanalytic time-space tradeoffs:
compression arguments, the pre-sampling technique, and concentration inequalities.  
The pre-sampling technique and the concentration inequalities technique are closely related, as demonstrated in~\cite{GuoLLZ21}. Our techniques are related to both, as they are also based on concentration inequalities. Yet, the previously used concentration inequalities are somewhat similar to martingales, which are generally not strong enough to obtain ``dimension-free'' inequalities such as those obtained from Shearer's inequality (see~\cite[Section 4]{Vondr10}). 
Furthermore, it seems that compression arguments are not sufficiently strong to obtain our results as well (though, we cannot give a short intuitive argument for this). 

We now argue in more detail that previous techniques are inherently limited: even when applied to non-adaptive algorithms, they can improve over the tight bounds known for adaptive algorithms by at most a polylogarithmic factor in $N$. Thus, they cannot meaningfully distinguish between adaptive and non-adaptive algorithms in our setting. For simplicity, we focus on the DLOG problem.
Since this argument is obvious for the pre-sampling technique,
we focus on techniques 
based on 
concentration inequalities and compression arguments.

We remark that we cannot completely rule out the possibility that 
a simple extension of these techniques would give a significant improvement,
but this seems highly unlikely.

\subsection{Concentration inequality-based proofs in the multi-instance model}

\subsubsection*{The MI model.} We summarize the main ideas of the technique that is based on concentration inequalities.
It reduces proving security against adversaries with advice to analyzing multi-instance (MI) security against adversaries without advice.
In the (basic) MI DLOG game, an adversary plays a sequence of DLOG instances.
Initially, a uniform permutation is chosen, and it remains fixed throughout the MI DLOG game.
In each instance, a new uniform DLOG secret is independently chosen and the adversary is allowed to issue $T$ new queries
in order to solve it. The adversary wins the multi-instance game if all instances are solved correctly.
The main theorem asserts that if the success probability (of any adversary) in 
the MI DLOG game with $S$ instances is at most $\delta^{S}$ for some $0 < \delta < 1$,
then the success probability of any adversary with $S$ bits of advice is at most $O(\delta)$.

For simplicity, we focus on the setting of a constant $\delta$.
Below we describe a non-adaptive adversary to the MI game with success probability of at least
$2^{-\tilde{O}(N/T^2)}$
regardless of the number of instances.
Asymptotically, this is also the success probability of the best general (adaptive) adversary, hence this technique cannot distinguish between adaptive and non-adaptive adversaries. 

Considering time-space tradeoffs in the proprocessing model, 
our adversary shows that this technique cannot prove a bound that is better than
$S T^2 \geq \tilde{\Omega}(N)$
(which is tight for adaptive algorithms \cite{Corrigan-GibbsK18}), while Theorem~\ref{thm:dl-nonadaptive-intro} proves a much better bound.

\subsubsection*{The MI adversary.} 
Assume that $T^2 = o(N)$, as otherwise, a non-preprocessing algorithm already achieves a constant advantage.
In addition, the MI adversary can guess each secret DLOG with probability $N^{-1}$,
hence it can achieve a success probability of at least
$2^{-\tilde{O}(N/T^2)}$
with $\tilde{O}(N/T^2)$ instances.  
We may thus assume that the number of instances $S$ is at least 
$\tilde{\Omega}(N/T^2)$.

In the MI game, we have a sequence of DLOG instances, with corresponding secrets $d_1,d_2,\ldots$.
Consider an MI adversary $B$ that for each instance, executes the same sequence of queries $(a_j,0)_{j\in [T]}$. We select a generator $g \in \mathbb{Z}_N^{*}$ and choose $a_i = g^{-i} \bmod N$ for $1 \leq i \leq T/2$
and $a_i = g^{(i - T/2) \cdot T} \bmod N$ for $T/2 < i \leq T$ (somewhat mimicking the baby-step-giant-step algorithm).
These queries are translated to oracle queries of the form $a_j d_i \bmod N$.

Observe that if a collision in the translated queries of the form 
$a_j d_i \bmod N = a_{j'} d_{i'}\bmod N$
occurs for $d_i \neq d_{i'}$,
then 
$d_i = a_{j'} \cdot (a_j)^{-1} \cdot d_{i'} \bmod N$
(recall that $a_j \neq 0$ by assumption).
Hence, $d_{i'}$ determines $d_i$ and vice versa.
Such collisions can be detected by $B$,
as they lead to collisions at the output of $\sigma$.

Thus, for each instance, if its (translated) queries do not collide with the (translated) queries of a previous instance, $B$ guesses its secret. Otherwise, using the collision, $B$ computes the secret of the instance deterministically using the previous guess.

\subsubsection*{Analysis sketch.} 
Let us show that for $S \geq \tilde{\Omega}(N/T^2)$, the success probability of $B$ is at least $2^{-\tilde{O}(N/T^2)}$. $B$ guesses $d_i$ correctly for all $1 \le i \le \tilde{O}(N/T^2)$ with probability $2^{-\tilde{O}(N/T^2)}$. In this case, $B$ succeeds if each of the remaining secrets can be determined from them. Hence, it remains to show that each of the remaining secrets can be determined with high probability.

Notably, after $\tilde{\Omega}(N/T^2)$ instances --- corresponding to a total of $\tilde{\Omega}(N/T)$ queries --- the following property holds with very high probability due to our choice of non-adaptive queries: every interval of the form $g^{i},g^{i+1},\ldots,g^{j}$, where $i,j \in [N-1]$ and $j - i \bmod N - 1 = T/2$, contains at least one queried point. The proof is by a Chernoff bound for each interval (exploiting the fact that each interval is hit by a query of an instance with probability $\Omega(T^2/N)$), followed by a union bound over all such intervals. Once this property holds, the queries of every new instance are guaranteed to collide with the queries of a (successfully guessed) previous one, allowing the recovery of its secret. 

\subsection{Compression arguments}
Compression arguments for proving time-space lower bounds for the DLOG problem have the following structure: one starts with a DLOG adversary with advice $A = (A_0, A_1)$ that possesses a (too good to be true) time-space tradeoff. This adversary is then used to devise encoding and decoding procedures for the random permutation oracle. These procedures compress the oracle beyond what is possible information-theoretically, leading to a contradiction.

Specifically, on input $\sigma$, the encoding procedure first runs $A_0(\sigma)$
and writes the advice $z$ of length $S$ bits as part of the encoding string.
Then (somewhat similarly to the MI game), 
it runs $A_1$ sequentially using $z$ on multiple DLOG instances,
and the goal is to gain beyond $S$ bits in the encoding length, 
establishing a contradiction.\footnote{These DLOG instances have non-uniformly distributed secrets, on which the success of $A_1$ is not guaranteed. However, this problem is addressed using the random self-reducibility property of the DLOG problem, which allows to execute $A_1$ on instances where the secret is uniformly distributed.}

In each instance, the encoding procedure encodes the query answers of $A_1$, so that the decoding procedure
can repeat the same steps.
Assuming that $A_1$ is deterministic, this is sufficient (otherwise, shared randomness is used).  
A saving in the encoding length is obtained when $A_1$ answers the instance correctly,
as this answer gives information about $\sigma$ (of at most $\log N$ bits) 
that does not need to be encoded.  

For simplicity, we assume that $A$ answers correctly with constant probability. 
We shall show that after answering $\tilde{\Omega}(N/T^2)$ instances,
a saving cannot be obtained (with high probability).
Thus, the total saving is bounded by $\tilde{O}(N/T^2)$ bits
and the technique cannot be used to prove a time-space tradeoff that is better than
$S T^2 \geq \tilde{\Omega}(N)$.
Once again, this time-space tradeoff is tight for adaptive algorithms,
and this technique cannot distinguish between adaptive and non-adaptive algorithms (up to logarithmic factors).

It remains to show that after answering $\tilde{\Omega}(N/T^2)$ instances,
a saving cannot be obtained (with high probability). 
The argument here is similar to the one used for concentration inequalities.
At the threshold of $\tilde{\Omega}(N/T^2)$ instances, about $\tilde{\Omega}(N/T)$ queries were already made. Therefore, it is very likely that at least one of the $T$ queries for the new instance will collide with a previous one.
Such a collision query gives no new information about $\sigma$,
hence encoding it in a standard way, using about $\log N$ bits, results in a loss that
nullifies the encoding advantage from running $A_1$.

One may try to encode collisions in a special way to mitigate the loss
by encoding the specific query indices that collide.
However, 
at the threshold of $\tilde{\Omega}(N/T^2)$ instances
(where $\tilde{\Omega}(N/T)$ queries have already been issued),
encoding the query pair seems to require more than
 $\log T + (\log N - \log T) = \log N$ bits,
once again nullifying the encoding advantage from running $A_1$.

\section{Non-adaptive sqDDH Algorithm}
\label{app:sqddh}
We sketch the details of the non-adaptive variant of the adaptive sqDDH algorithm of~\cite{Corrigan-GibbsK18}.
Essentially, the algorithm replaces the adaptive random walk of~\cite{Corrigan-GibbsK18} with a sequence that can be computed non-adaptively. 

First, define a biased pseudorandom predicate 
$\textup{P}:[N]^2 \to \{0,1\}$
that evaluates to $1$ with probability about $T^{-1}$.
We call a pair of the form
$(\sigma(x),\sigma(x^2))$ (where $x \in [N]$) special if
$\textup{P}(\sigma(x),\sigma(x^2)) = 1$.

Second, define a balanced pseudorandom predicate  $\textup{Q}:[N]^2 \to \{0,1\}$
that evaluates to $1$ with probability $1/2$.

Third, define a pseudorandom function 
$g:[N]^2 \to [S]$ that evaluates to each $i \in [S]$ with probability about $S^{-1}$.

Given $\sigma$, for each $v \in [S]$, define the set 
$$\mathcal{L}_v = \{(\sigma(x),\sigma(x^2)) \in [N]^2 \colon \textup{P}(\sigma(x),\sigma(x^2)) = 1 \wedge g(\sigma(x),\sigma(x^2)) = v\},$$
which contains roughly $\frac{N}{ST}$ pairs.
Moreover, define the majority value of 
$\textup{Q}$ on $\mathcal{L}_v$ by 
$$\textup{Maj}_v = 
\textup{Majority}\{\textup{Q}(\sigma(x),\sigma(x^2)) \colon (\sigma(x),\sigma(x^2)) \in \mathcal{L}_v\}.$$
A standard probabilistic argument shows that with high probability, $\textup{Maj}_v$ agrees with a pair chosen uniformly form $\mathcal{L}_v$ on the value of $\textup{Q}$  with probability about $\frac{1}{2} + \sqrt{\frac{ST}{N}}$.
Namely, with high probability,
$$
\Pr_{(\sigma(x),\sigma(x^2)) \sim \mathcal{L}_v}[\textup{Q}(\sigma(x),\sigma(x^2)) = \textup{Maj}_v] \geq \frac{1}{2} + \Omega \Bigl(\sqrt{\frac{ST}{N}} \Bigr).
$$

The preprocessing algorithm $A_0$ works by 
calculating and storing 
$\textup{Maj}_v$ for each $v \in [S]$ in the advice string.

The online algorithm $A_1$ defines a sequence of $T/2$ pairs in $[N]^2$ that can be computed using non-adaptive queries and preserve the sqDDH relation. For example, 
for some pseudorandom function $f: [N] \rightarrow [N]$, starting from the secret $(d_1,d_2)$, define the sequence $(\sigma(d_1),\sigma(d_2)),(\sigma(d_1 \cdot f(1)),\sigma(d_2 \cdot f(1)^2)),(\sigma(d_1 \cdot f(2)),\sigma(d_2 \cdot f(2)^2)) \ldots$.

With high probability, one of these $T/2$ pairs evaluates to $1$ on $\textup{P}$ (otherwise, $A_1$ guesses the answer).
Denote the first such pair by
$(y_1,y_2)$. 
$A_1$ computes 
$v:= g(y_1,y_2)$ and 
$b := \textup{Q}(y_1,y_2)$.
It then returns $1$ if 
$b = \textup{Maj}_v$,
and $0$ otherwise. 

The main observation in the analysis is
that on a YES instance, 
$b = \textup{Maj}_v$ with probability about 
$\frac{1}{2} + \sqrt{\frac{ST}{N}}$ (as noted above). On the other hand, on a NO instance, $b = \textup{Maj}_v$ with probability of about $\frac{1}{2}$, since 
$\textup{Q}$ is balanced and uncorrelated with pairs computed for NO instances.

\section{A Variant of Shearer's Lemma for Bijections}
\label{app:shearer-perm}

In this appendix we consider Theorem~\ref{thm:shearer-perm-intro} -- namely, the variant of Shearer's lemma for bijections. First, we prove that the theorem follows from~\cite[Proposition~21]{BartheCLM11} of Barthe, Cordero-Erausquin, Ledoux and Maurey and from~\cite[Theorem~4]{CaputoS24} of Caputo and Salez.
Then we present an elementary proof of a slightly weaker version of the theorem in which the constant $2$ is replaced by $9$.

\subsection{Deduction of Theorem~\ref{thm:shearer-perm-intro} from the results of~\cite{BartheCLM11,CaputoS24}}
\label{sec:sub:Deduction}

In this subsection we show how Theorem~\ref{thm:shearer-perm-intro} follows from~\cite[Proposition~21]{BartheCLM11} and from~\cite[Theorem~4]{CaputoS24}, which are stated in a very different form. Let us recall Theorem~\ref{thm:shearer-perm-intro}.

\medskip \noindent \textbf{Theorem~\ref{thm:shearer-perm-intro}.} Let $\mathcal{X}$ be a set of size $N$. Let $Q_X = Q_{X_1,\ldots,X_N}$ be the uniform distribution over bijections from $[N]$ to $\mathcal{X}$,
and let $P_X = P_{X_1,\ldots,X_N}$ be another distribution over such bijections. Let 
$\Qcal_1,\Qcal_2,\ldots,\Qcal_m$ be subsets of $[N]$, such that each $i \in [N]$ belongs to at most $k$ of them. Then
$$
2k \cdot  \KL( P_X \| Q_X ) \geq \sum_{j \in [m]} \KL( P_{X_{\mathcal{U}_j}} \| Q_{X_{\mathcal{U}_j}} ),$$
where $P_{X_{\mathcal{U}}}$ is the distribution of the vector $X_{\mathcal{U}} := \left(X_i \mid i \in \mathcal{U} \right)$ with respect to $P$ (and analogously for $Q$).

\medskip \noindent Since the equivalence between~\cite[Proposition~21]{BartheCLM11} and~\cite[Theorem~4]{CaputoS24} is discussed in a remark after~\cite[Theorem~4]{CaputoS24}, we show how Theorem~\ref{thm:shearer-perm-intro} follows from~\cite[Theorem~4]{CaputoS24}. We first cite this result and then explain the notations involved.
\begin{theorem}[{\cite[Theorem~4, Equation (34)]{CaputoS24}}]\label{thm:shearer perm caputo}
For any choice of a probability vector $(\theta_A, A \subseteq [N])$ on subsets of $[N]$ and every function $f:S_N \to \RR_{\ge 0}$,
\[
    \sum_{A \subseteq [N]} \theta_A\Ent(\text{E}_A[f]) \le (1-\kappa)\Ent(f),
\]
where $\kappa = \min_{i \ne j} \sum_{A \supseteq \{i,j\}} \theta_A$.
\end{theorem}
The reader is referred to \cite{CaputoS24} for a full description of the functional $\Ent$ and its properties. Here we only state two basic observations, that follow immediately from its definitions:
\begin{observation}\label{obs:ent VS KL}
Let $\mathcal{X}$ be a set of size $N$. Let $Q_X = Q_{X_1,\ldots,X_N}$ be the uniform distribution over bijections from $[N]$ to $\mathcal{X}$, let $P_X = P_{X_1,\ldots,X_N}$ be another distribution over such bijections, and let $f:S_N \to \RR$ be defined as $f(\sigma) = \frac{P(\sigma)}{Q(\sigma)}$. Let $A \subseteq [N]$ be a set of indices. Then:
\begin{itemize}
    \item $\Ent(f) = \KL(P \| Q)$.
    \item $\Ent(\text{E}_A[f]) = \KL( P_{X_{A^c}} \| Q_{X_{A^c}} )$ (note that the distributions are projected to the complement set $A^c = [N] \setminus A$ and not to $A$).
\end{itemize}
\end{observation}
\noindent Using these observations, we show that Theorem~\ref{thm:shearer-perm-intro} follows directly from Theorem~\ref{thm:shearer perm caputo}.

\medskip 

\begin{proof}[of Theorem~\ref{thm:shearer-perm-intro}]
Let $P,Q,f$ be as in Observation~\ref{obs:ent VS KL}. Let $\Qcal_1,\Qcal_2,\ldots,\Qcal_m$ be subsets of $[N]$, such that each $i \in [N]$ belongs to at most $k$ of them. For every $A \subseteq [N]$, we denote $\theta_A = |\{j \mid \Qcal_j = A^c\}|/m$. Applying Theorem~\ref{thm:shearer perm caputo} we obtain
\[
    \sum_{A \subseteq [N]} \theta_A\Ent(\text{E}_A[f]) \le (1-\kappa)\Ent(f),
\]
where $\kappa = \min_{i \ne i'} \sum_{A \supseteq \{i,i'\}} \theta_A$. By Observation~\ref{obs:ent VS KL}, an equivalent formula is
\[
    \sum_{A \subseteq [N]} \theta_A\KL( P_{X_{A^c}} \| Q_{X_{A^c}} ) \le (1-\kappa)\KL(P \| Q).
\]
Substituting the formula for $\theta_A$, we obtain
\begin{equation}\label{eqn:proof of ent VS kl}
    \frac{1}{m} \sum_{j=1}^m \KL( P_{X_{\Qcal_j}} \| Q_{X_{\Qcal_j}} )  \le (1-\kappa)\KL(P \| Q).
\end{equation}

It remains to bound $\kappa$. Notice that $\kappa = \min_{i \ne i'} \sum_{A \supseteq \{i,i'\}} \theta_A = \frac{1}{m} \min_{i \ne i'}|\{j\mid \{i,i'\} \subseteq\Qcal_j^c\}|$. From the assumption on $\Qcal_j$ we obtain that $|\{j\mid i \in\Qcal_j^c\}| \ge m-k$ for all $i$, implying that $\kappa \ge \frac{m-2k}{m}$. The statement now follows from~\eqref{eqn:proof of ent VS kl}.
\end{proof}

%\medskip \noindent An elementary direct proof of a weaker version of Theorem~\ref{thm:shearer-perm-intro}, with the constant $9$ instead of $2$, is presented in Appendix~\ref{app:shearer-perm}.

\subsection{An elementary proof of a slighly weaker version of the theorem}
\label{sec:sub:weaker-Shearer}

In this subsection we present an elementary proof of a slightly weaker version of Theorem~\ref{thm:shearer-perm-intro}, in which the constant $2$ is replaced by $9$. Specifically, we prove the following. 
\begin{theorem}\label{thm:Shearer-appendix}
    Let $\mathcal{X}$ be a set of size $N$. Let $Q_X = Q_{X_1,\ldots,X_N}$ be the uniform distribution over bijections from $[N]$ to $\mathcal{X}$,
    and let $P_X = P_{X_1,\ldots,X_N}$ be another distribution over such bijections. Let 
    $\Qcal_1,\Qcal_2,\ldots,\Qcal_m$ be subsets of $[N]$, such that each $i \in [N]$ belongs to at most $k$ of them. Then
    $$
    9k \cdot  \KL( P_X \| Q_X ) \geq \sum_{j \in [m]} \KL( P_{X_{\mathcal{U}_j}} \| Q_{X_{\mathcal{U}_j}} ),
    $$
    where $P_{X_{\mathcal{U}}}$ is the distribution of the vector $X_{\mathcal{U}} := \left(X_i \mid i \in \mathcal{U} \right)$ with respect to $P$ (and analogously for $Q$).
\end{theorem}

We need several tools in order to prove the theorem. Let us begin with the following technical lemma:
\begin{lemma}\label{lem:technical lemma}
Let $N \in \NN$ and let $\ell\le \frac{N}{4}$. Let $p_1,\dots,p_\ell \in (0,1)$ with $\sum_{i}p_{i}\le 1$, and let us denote $\epsilon_{i}=p_{i}-\frac{1}{N}$. Denote also $p'=\frac{1-\sum_{i}p_{i}}{N-\ell}$ and $\epsilon'=p'-\frac{1}{N}$. Then the following holds:
\[
    Np'\log\left(Np'\right)-N\epsilon' \le 4 \sum_{i}\left(p_{i}\log(Np_{i})-\epsilon_{i}\right).
\]
\end{lemma}
\begin{proof}
The proof utilizes the following inequalities (recall that logarithms are in base $e$):
\begin{enumerate}[label=(\roman*), ref=\text{ineq. (}\roman*\text{)}]
    \item \label{ineq nonnegative} For all $x>0$ it holds that $x\log x-x+1 \ge 0$.
    \item \label{ineq upper bound} For all $0 < x \le 2$ it holds that $x\log x-x+1\le(x-1)^{2}\le\left|x-1\right|$.
    \item \label{ineq at least square} For all $x<5$ it holds that $x\log x-x+1\ge\frac{1}{4}(x-1)^{2}$.
    \item \label{ineq at least abs} For all $x \ge 5$ it holds that $x\log x-x+1\ge\left|x-1\right|$.
\end{enumerate}

All these inequalities can be easily proved by basic analytical tools.

\medskip 

We prove the statement by case-analysis. Let us first assume that 
\[
    \sum_{i:Np_{i}\ge5}\left|\epsilon_{i}\right|\ge\frac{1}{2}\sum_{i:Np_{i}<5}\left|\epsilon_{i}\right|.
\]
We obtain 
\[
    \sum_{i}\left|\epsilon_{i}\right|\le 3 \sum_{i:Np_{i}\ge5}\left|\epsilon_{i}\right|.
\]
Since $\ell\le\frac{N}{4}$ and following the definition of $\epsilon'$, we may also obtain 
\begin{equation}\label{eqn:first case bound n epsilon}
    \left|N\epsilon'\right|\le\frac{4}{3}\sum_{i}\left|\epsilon_{i}\right| \le 4 \sum_{i:Np_{i} \ge 5} \left|\epsilon_{i}\right|.
\end{equation}
Therefore, we may bound:
\begin{align*}
    Np'\log\left(Np'\right)-N\epsilon' & =  (Np')\log(Np')-Np'+1 &\\
    & \le  \left|N\epsilon'\right| & \text{\ref{ineq upper bound} with $x=Np'$}\\
    & \le  4\sum_{i:Np_{i}\ge5}\left|\epsilon_{i}\right| & \text{Equation \eqref{eqn:first case bound n epsilon}}\\
    & \le  \frac{4}{N}\sum_{i:Np_{i}\ge5}\left(Np_{i}\log(Np_{i})-Np_{i}+1\right) & \text{\ref{ineq at least abs} with $x=Np_i$}\\
    & =  4\sum_{i:Np_{i}\ge5}\left(p_{i}\log(Np_{i})-\epsilon_{i}\right) &\\
    & \le  4\sum_{i}\left(p_{i}\log(Np_{i})-\epsilon_{i}\right) & \text{\ref{ineq nonnegative} with $x=Np_i$,}
\end{align*}
as required.

\medskip \noindent Next, we assume
\[
    \sum_{i:Np_{i}\ge5}\left|\epsilon_{i}\right|<\frac{1}{2}\sum_{i:Np_{i}<5}\left|\epsilon_{i}\right|,
\]
implying that
\begin{equation}\label{eqn:second case bound n epsilon}
    \left|N\epsilon'\right|\le\frac{4}{3}\sum_{i}\left|\epsilon_{i}\right| \le 2 \sum_{i:Np_{i}<5} \left|\epsilon_{i}\right|.
\end{equation}
Therefore we may bound:
\begin{align*}
    Np'\log\left(Np'\right)-N\epsilon' & =  (Np')\log(Np')-Np'+1 &\\
    & \le  \left(N\epsilon'\right)^2 & \text{\ref{ineq upper bound} with $x=Np'$}\\
    & \le  \left(2\sum_{i:Np_{i}<5}\left|\epsilon_{i}\right|\right)^2 & \text{Equation \eqref{eqn:second case bound n epsilon}}\\
    & \le  4\ell\sum_{i:Np_{i}<5}\epsilon_{i}^{2} & \text{Cauchy–Schwarz inequality}\\
    & \le  \frac{16\ell}{N}\sum_{i:Np_{i}<5}\left(p_{i}\log(N p_{i})-\epsilon_{i}\right) & \text{\ref{ineq at least square} with $x=Np_i$}\\
    & \le  4\sum_{i:Np_{i}<5}\left(p_{i}\log(N p_{i})-\epsilon_{i}\right) & 4\ell\le N\\
    & \le  4\sum_{i}\left(p_{i}\log(Np_{i})-\epsilon_{i}\right) & \text{\ref{ineq nonnegative} with $x=Np_i$,}
\end{align*}
as required.
\end{proof}

\medskip \noindent Lemma~\ref{lem:technical lemma} allows proving a variant of Shearer's inequality for indicator vectors.
\begin{lemma}[Variant of Shearer's inequality for indicator vectors] 
\label{lem:shearer-indicator}
Let $N$ be a positive integer. Let $Q_X = Q_{X_1,\ldots,X_N}$ be the uniform distribution function over the set $\ind{N} := \{v \in \{0,1\}^N \mid \sum_i v_i=1\}$ of indicator vectors,
and let $P_X = P_{X_1,\ldots,X_N}$ be another distribution function over the same set. Let $\mathcal{U}_1,\ldots,\mathcal{U}_m \subseteq [N]$ be such that for each 
$i \in [N]$, $|\{ j \in [m] \mid i \in \mathcal{U}_{j} \} | \leq k$.
Then
$$9k \cdot  \KL( P_X \| Q_X ) \geq \sum_{j \in [m]} \KL( P_{X_{\mathcal{U}_j}} \| Q_{X_{\mathcal{U}_j}} ).$$
\end{lemma}

\begin{proof}
Denote $S = \KL( P_X \| Q_X )$. By applying \ref{prope:data processing} we obtain that for every $\mathcal{U} \subseteq [N]$ we have $\KL( P_{X_{\mathcal{U}}} \| Q_{X_{\mathcal{U}}} ) \le S$. Since the number of sets $\mathcal{U}_j$ with $|\mathcal{U}_j| > \frac{N}{4}$ is at most $4k$, we may assume that $|\mathcal{U}_j| \le \frac{N}{4}$ and prove that
\[
    \sum_{j \in [m]} \KL( P_{X_{\mathcal{U}_j}} \| Q_{X_{\mathcal{U}_j}} ) \le 5kS.
\]
Moreover, since $\KL( P_{X_i} \| Q_{X_i} ) \ge 0$ following \ref{prope:kl nonnegative}, we may assume that $|\{ j \in [m] \mid i \in \mathcal{U}_{j} \} | = k$ for all $i \in [N]$.

For every $i \in [N]$, we notice that $Q_{X_i}(1) = \frac{1}{N}$ (i.e., the probability of the event $X_i=1$ when sampled with respect to $Q$ is $1/N$) and denote $p_i = P_{X_i}(1)$ and $\epsilon_i = p_i - \frac{1}{N}$, to obtain by definition $S = \sum_{i}p_{i}\log(Np_{i})$.
Additionally, for every set $\mathcal{U} \subseteq [N]$ we denote
\[
    p'(\mathcal{U})=\frac{1-\sum_{i\in \mathcal{U}}p_{i}}{N-|\mathcal{U}|} \quad \text{and}\quad\epsilon'(\mathcal{U})=p'(\mathcal{U})-\frac{1}{N} = - \frac{\sum_{i\in \mathcal{U}} \epsilon_i}{N-|\mathcal{U}|}.
\]
Let $\mathcal{U} \subseteq [N]$ be an index set. By the definition of KL-divergence, we may obtain:
\begin{align*}
    \KL( P_{X_{\mathcal{U}}} \| Q_{X_{\mathcal{U}}} ) &	= \sum_{v\in \Supp(X_{\mathcal{U}})} P_{X_{\mathcal{U}}}(v)\cdot\log\left(\frac{P_{X_{\mathcal{U}}}(v)}{Q_{X_{\mathcal{U}}}(v)}\right)\\
    & = \sum_{i\in \mathcal{U}}p_{i}\log\left(Np_{i}\right) + \left(1-\sum_{i\in \mathcal{U}}p_{i} \right) \log\left(\frac{1-\sum_{i\in \mathcal{U}}p_{i}}{1-\frac{|\mathcal{U}|}{N}}\right)\\
    & =  \sum_{i\in \mathcal{U}}p_{i}\log\left(Np_{i}\right) + (N-|\mathcal{U}|) p'(\mathcal{U})\log\left(Np'(\mathcal{U})\right).
\end{align*}
Notably, $\sum_{j \in [m]}\sum_{i\in \mathcal{U}_j}p_{i}\log\left(Np_{i}\right) = kS$, so it remains to bound $$\sum_{j \in [m]}(N-|\mathcal{U}_j|)p'(\mathcal{U}_j)\log\left(Np'(\mathcal{U}_j)\right).$$
Since $|\{ j \in [m] \mid i \in \mathcal{U}_{j} \} | = k$ is assumed to be constant, we obtain that 
\[
    \sum_{j \in [m]}(N-|\mathcal{U}_j|)\epsilon'(\mathcal{U}_j) = -\sum_{i} k\epsilon_i=0.
\]
Therefore, we may bound:
\begin{align*}
    \sum_{j \in [m]}(N-|\mathcal{U}_j|)p'(\mathcal{U}_j)\log\left(Np'(\mathcal{U}_j)\right) & = \sum_{j \in [m]}(N-|\mathcal{U}_j|)\left(p'(\mathcal{U}_j)\log\left(Np'(\mathcal{U}_j)\right) - \epsilon'(\mathcal{U}_j)\right)\\
    & \le \sum_{j \in [m]}N\left(p'(\mathcal{U}_j)\log\left(Np'(\mathcal{U}_j)\right) - \epsilon'(\mathcal{U}_j)\right)\\
    & \le 4\sum_{j \in [m]}\sum_{i \in \mathcal{U}_j}\left(p_{i}\log(Np_{i})-\epsilon_{i}\right)\\
    & = 4kS,
\end{align*}
where the first inequality follows from the fact that $\left(p'(\mathcal{U}_j)\log\left(Np'(\mathcal{U}_j)\right) - \epsilon'(\mathcal{U}_j)\right) \ge 0$ (implied by \ref{ineq nonnegative} from the proof of Lemma~\ref{lem:technical lemma}), and the second inequality follows from Lemma~\ref{lem:technical lemma}.
\end{proof}

In order to prove Theorem~\ref{thm:Shearer-appendix}, we need another lemma which shows that under an additional assumption, conditioning does not decrease KL-divergence.

\begin{lemma}\label{lem:conditinal KL bound}
Let $P$ and $Q$ be two distributions on pairs $(X,Y)$. Suppose that $X$ and $Y$ are independent with respect to $Q$ (i.e., $Q_{X,Y} = Q_X Q_Y$). Then
\[
    \KL( P_{Y} \| Q_{Y}) \le \KL( P_{Y \mid X} \| Q_{Y \mid X} ).
\]
\end{lemma}
\begin{proof}
Since $Q_{X,Y} = Q_X Q_Y$, we obtain
\[
    Q_Y(y) = \sum_x P_X(x)Q_{Y\mid X=x}(y).
\]
Moreover, by the law of total probability, we obtain
\[
    P_Y(y) = \sum_x P_X(x)P_{Y\mid X=x}(y).
\]
Therefore, we obtain
\[
    \KL( P_{Y} \| Q_{Y}) = \KL\left(\left(\sum_x P_X(x)P_{Y\mid X=x}(y)\right) \| \left(\sum_x P_X(x)Q_{Y\mid X=x}(y)\right)\right).
\]
Finally, we may apply \ref{prope:kl convex} inductively to obtain
\[
    \KL( P_{Y} \| Q_{Y}) \le \sum_x P_X(x)\KL(P_{Y\mid X=x} \| Q_{Y\mid X=x}) =: \KL( P_{Y \mid X} \| Q_{Y \mid X} ).
\]
\end{proof}

\medskip \noindent Now we are ready to prove Theorem~\ref{thm:Shearer-appendix}.

\medskip 

\begin{proof}[Proof of Theorem~\ref{thm:Shearer-appendix}]
First we set some notations, given a bijection $X = X_1,\dots,X_N:[N] \to \mathcal{X}$, indices $i,j\in [N]$ and a set $\mathcal{S} \subseteq [N]$:
\begin{eqnarray*}
    X_{i\to j}	&=&	\begin{cases}
                        1 & X_{i}=j,\\
                        0 & \text{otherwise;}
                    \end{cases}\\
    X_{\to j}	&=&	X^{-1}(j),\\
    X_{\mathcal{S}\to j}	&=&	\begin{cases}
                        X^{-1}(j) & X^{-1}(j)\in \mathcal{S},\\
                        * & \text{otherwise.}
                    \end{cases}
\end{eqnarray*}

Note that $X = X_1, \dots, X_N$ and the sequence $X_{\to 1}, \dots, X_{\to N}$ are equivalent in the sense that they both encode the same information about the permutation. Similarly, $X_i$ is equivalent to $(X_{i\to1},\dots,X_{i\to N})$, and $X_\mathcal{S}$ is equivalent to $(X_{\mathcal{S}\to1},\dots,X_{\mathcal{S}\to N})$.

By \ref{prope:chain rule kl}, we obtain:
\begin{equation}\label{eqn:shearer perm---general chain}
    \KL( P_X \| Q_X ) = \KL( P_{X_{\to1}} \| Q_{X_{\to1}} ) + \KL\left(P_{X_{\to 2}, \dots, X_{\to N} \mid X_{\to 1}} \| Q_{X_{\to 2}, \dots, X_{\to N} \mid X_{\to 1}}\right).
\end{equation}
Similarly, given a set $\mathcal{U}$, we obtain:
\begin{align}\label{eqn:shearer perm---set chain}
\begin{split}
    \KL( P_{X_\mathcal{U}} \| Q_{X_\mathcal{U}} ) = &\KL( P_{X_{\mathcal{U}\to1}} \| Q_{X_{\mathcal{U}\to1}} ) +\\ &+ \KL\left(P_{X_{\mathcal{U}\to 2}, \dots, X_{\mathcal{U}\to N} \mid X_{\mathcal{U}\to 1}} \| Q_{X_{\mathcal{U}\to 2}, \dots, X_{\mathcal{U}\to N} \mid X_{\mathcal{U}\to 1}}\right).
\end{split}
\end{align}

By Lemma~\ref{lem:shearer-indicator}, we have
\begin{equation}\label{eqn:shearer perm---indicator}
    \sum_{j \in [m]} \KL( P_{X_{\mathcal{U}_j\to1}} \| Q_{X_{\mathcal{U}_j\to1}} ) \le 9k\cdot \KL( P_{X_{\to1}} \| Q_{X_{\to1}} ).
\end{equation}

In addition, notice that the conditioned variable $X_{\mathcal{U}\to 2}, \dots, X_{\mathcal{U}\to N} \mid X_{\mathcal{U}\to 1}$, distributed with respect to $Q$, is independent of $X_{\to1}$. Indeed, if $X_{\mathcal{U}\to 1} = i \in \mathcal{U}$ then $X_{\to1} = i$ deterministically. Otherwise, $X_{\mathcal{U}\to 1} = *$, implying that $X_{\mathcal{U}}$ is a uniform sequence of $|\mathcal{U}|$ distinct elements other than $1$ (since $Q$ is the uniform distribution on bijections). For any $i \notin \mathcal{U}$, conditioning on $X_{\to 1} = i$ has no effect on this distribution. Therefore, we may apply Lemma~\ref{lem:conditinal KL bound} to obtain:
\[
    \KL\left(P_{X_{\mathcal{U}\to 2}, \dots, X_{\mathcal{U}\to N} \mid X_{\mathcal{U}\to 1}} \| Q_{X_{\mathcal{U}\to 2}, \dots, X_{\mathcal{U}\to N} \mid X_{\mathcal{U}\to 1}}\right) \le \KL\left(P_{X_{\mathcal{U}\to 2}, \dots, X_{\mathcal{U}\to N} \mid X_{\to 1}} \| Q_{X_{\mathcal{U}\to 2}, \dots, X_{\mathcal{U}\to N} \mid X_{\to 1}}\right).
\]
The distribution $Q_{X_{\mathcal{U}\to 2}, \dots, X_{\mathcal{U}\to N} \mid X_{\to 1}}$ is equivalent to a uniform distribution over bijections on $N-1$ elements. Therefore, we may apply an inductive argument to obtain
\begin{align*}
    \sum_{j \in [m]} \KL &\left(P_{X_{\mathcal{U}_j\to 2}, \dots, X_{\mathcal{U}_j\to N} \mid X_{\mathcal{U}_j\to 1}} \| Q_{X_{\mathcal{U}_j\to 2}, \dots, X_{\mathcal{U}_j\to N} \mid X_{\mathcal{U}_j\to 1}}\right)\\
    \le & \sum_{j \in [m]} \KL\left(P_{X_{\mathcal{U}_j\to 2}, \dots, X_{\mathcal{U}_j\to N} \mid X_{\to 1}} \| Q_{X_{\mathcal{U}_j\to 2}, \dots, X_{\mathcal{U}_j\to N} \mid X_{\to 1}}\right)\\
    \le & 9k\cdot \KL\left(P_{X_{\to 2}, \dots, X_{\to N} \mid X_{\to 1}} \| Q_{X_{\to 2}, \dots, X_{\to N} \mid X_{\to 1}}\right).
\end{align*}
Combining with Equations \eqref{eqn:shearer perm---general chain}, \eqref{eqn:shearer perm---set chain}, and \eqref{eqn:shearer perm---indicator}, we obtain
\[
    \sum_{j \in [m]} \KL( P_{X_{\mathcal{U}_j}} \| Q_{X_{\mathcal{U}_j}} ) \le 9k \cdot  \KL( P_X \| Q_X ),
\]
completing the proof.
\end{proof}

\section{Improved Bound for Problems without Post-Processing}
\label{app:improvement no post process}
Recall Theorem~\ref{thm:g-nonadaptive}. A slightly better bound can be achieved for problems with a trivial post-processing function $\POST(\secv,j)=j$. This rule captures the DLOG and DDH problems, but does not capture EM key recovery.

\begin{theorem} \label{thm:g-nonadaptive no post}
Let $\PC := \PC(N,\SECS,\mathcal{M},\TR,\POST,\SUC)$ be a permutation challenge game, such that 
$\TR$ is $u$-uniform and $\POST(\secv,j)=j$ is a trivial post-processing function.
Let $(A_0,A_1)$ be an $(S,T)$ non-adaptive algorithm with preprocessing for $\PC$.
Denote by $\MAXS(T)$ the optimal success probability (with respect to $\SUC$) of a non-preprocessing, non-adaptive algorithm that makes at most $T$ queries.
Then, the success probability of $(A_0,A_1)$ is at most
\[
\min \left(2\cdot \MAXS(T) + \frac{4\log(2) S T }{u}, \MAXS(T) + 
\sqrt{\frac{\log(2) ST} {{u}}} \right).
\]
\end{theorem}

\begin{proof}
By Theorem~\ref{thm:pc to middle}, the success probability of $(A_0,A_1)$ is at most
\[
\min \left(2\cdot \widehat{\MAXS}(T) + \frac{4\log(2) S T }{u}, \widehat{\MAXS}(T) + 
\sqrt{\frac{\log(2) ST} {{u}}}\right),
\]
where $\widehat{\MAXS}(T)$ is the optimal success probability of an adversary to $\MID$ with running time $T$.
Thus, it remains to show that $\widehat{\MAXS}(T)\le \MAXS(T)$ and complete the proof. Let $\widehat{A} = (\widehat{A}_0,\widehat{A}_1)$ be an adversary to $\MID$ with success probability $\widehat{\MAXS}(T)$. Let us describe a non-preprocess algorithm $A$ to the problem:
\begin{enumerate}
    \item $A$ receives $N$ and a non-adaptive oracle access by inner and outer queries to the permutation $\sigma$ with respect to a secret $\secv$.
    \item $A$ runs $\widehat{A}_0(N)$ to obtain the sequences $\INP = (\INP_1,\dots,\INP_{T_1})$ and $\OUTP = (\OUTP_1,\dots,\OUTP_{T_1})$.
    \item $A$ runs $\widehat{A}_1$ and obtains its queries (recall that $\widehat{A}_1$ is allowed to make only outer queries), $m_1,\dots,m_{T_2} \in \mathcal{M}$.
    \item $A$ makes the inner queries $\INP_1,\dots,\INP_{T_1}$ and obtains $\OUTP'_1,\dots,\OUTP'_{T_1}$. Additionally, $A$ makes the outer queries $m_1,\dots,m_{T_2}$. Since the post-processing function is assumed to be trivial, the outputs are of the form $r_i := \sigma(\TR(\secv,m_i))$.
    \item Using some deterministic process that depends only on $\{\OUTP_1,\dots,\OUTP_{T_1},\OUTP'_1,\dots,\OUTP'_{T_1}\}$, $A$ constructs a permutation $\pi$ with $\pi(\OUTP'_i) = \OUTP_i$ for all $i$. 
    \item $A$ simulates $\widehat{A}_1$'s oracle, and returns $\pi(r_i)$ as the response to the query $m_i$. $A$ returns $\widehat{A}_1$'s output.
\end{enumerate}

Clearly, the permutation $\pi \circ \sigma$ is uniformly distributed over the permutations with $\forall i:\INP_i \mapsto \OUTP_i$. Moreover, the responses $\widehat{A}_1$ obtained are $\pi(r_i) = (\pi \circ \sigma)(\TR(\secv,m_i))$. Therefore, the distribution of the permutation $\pi \circ \sigma$ is exactly as it is supposed to be, and the responses to all queries are correct with respect to it. Therefore, the simulation is accurate and does not affect the success probability. Thus, the success probability of $A$ is exactly $\widehat{\MAXS}(T)$.
\end{proof}

\subsubsection*{Acknowledgments.} We thank the anonymous reviewers for their useful comments. In particular, we thank a reviewer who pointed out an error in the proof of Theorem~\ref{thm:pc to middle} in a previous version of this paper.  

\bibliographystyle{alpha}	
\bibliography{proposal/shearer-perm}	

\end{document}